\documentclass[11pt]{article}

\usepackage[a4paper,margin=1in]{geometry}
\usepackage{amsmath,amssymb,amsthm,mathtools,bm}
\usepackage{booktabs,array,multirow}
\usepackage{graphicx,subcaption}
\usepackage{microtype}
\usepackage{enumitem}
\usepackage{hyperref}
\usepackage{float,placeins}
\usepackage{tikz}
\usetikzlibrary{arrows.meta,positioning}
\hypersetup{colorlinks=true,linkcolor=blue,citecolor=blue,urlcolor=blue}
\graphicspath{{figures/}}
\allowdisplaybreaks
\newtheorem{theorem}{Theorem}[section]
\newtheorem{proposition}[theorem]{Proposition}
\newtheorem{corollary}[theorem]{Corollary}
\newtheorem{lemma}[theorem]{Lemma}
\theoremstyle{definition}

\theoremstyle{remark}
\newtheorem{remark}[theorem]{Remark}

\newcommand{\E}{\mathbb E}
\newcommand{\Pp}{\mathbb P}
\newcommand{\R}{\mathbb R}
\newcommand{\N}{\mathbb N}
\newcommand{\Unif}{\operatorname{Unif}}
\newcommand{\op}{\operatorname{op}}
\newcommand{\TV}{\operatorname{TV}}

\newcommand{\bxi}{\boldsymbol{\xi}}
\newcommand{\bzeta}{\boldsymbol{\zeta}}
\newcommand{\X}{\mathcal X}
\newcommand{\dd}{\,\mathrm d}

\title{Hankel--Christoffel--Nevai Screening for Bayesian Inverse Problems}
\author{Zhiliang Deng\thanks{Corresponding author. School of Mathematical Sciences, University of Electronic Science and Technology of China. Email: dengzhl@uestc.edu.cn.}
\and Xiaomei Yang\thanks{School of Mathematics, Southwest Jiaotong University. Email: yangxiaomath@swjtu.edu.cn.}}
\date{}

\begin{document}
\maketitle

\begin{abstract}
Likelihood evaluation in Bayesian inverse problems often requires a forward-model solve. We study candidate screening using a likelihood-weighted prior moment matrix and two associated scores: a Christoffel ratio and a Nevai polynomial average. A finite pilot supplies likelihood information that is reused to rank further prior candidates. We connect finite-pilot matrix error, quantitative Legendre localization, and fixed-budget posterior-mass regret. Conditional likelihoods distinguish feature loss from polynomial and sampling errors, and a filtered Gaussian construction permits controlled feature selection and localization. In a nonlinear function-coefficient PDE, equal-budget comparisons separate methods using scalar likelihoods from surrogates using forward outputs. The Christoffel ratio outperforms the tested direct likelihood regressions, but forward-response surrogates capture more posterior mass at small retention budgets. Independent-pilot tests and reference resampling quantify two distinct sources of uncertainty. Controlled Gaussian experiments illustrate finite-filter truncation and the benefit of selecting relevant features. These results identify when moment geometry provides useful screening information, without asserting universal superiority over surrogate models or an exact posterior sampling method.
\end{abstract}

\noindent\textbf{Keywords.}
Bayesian inverse problems; uncertainty quantification; posterior screening; likelihood-weighted moments; Christoffel function; Nevai score; function-valued parameter.

\section{Introduction}

Bayesian inverse problems update a prior measure by reweighting it
according to agreement with observed data
\cite{KaipioSomersalo2005,Stuart2010}.
Computing posterior expectations generally requires numerical methods
such as Markov chain Monte Carlo (MCMC), whose cost can be dominated
by repeated evaluations of an expensive forward model.
This motivates methods that use information about the posterior
to guide exploration.
Likelihood-informed subspace methods identify directions along which
the data most strongly modify the prior
\cite{CuiMarzoukWillcox2016}, and dimension-independent
likelihood-informed MCMC incorporates this information into
posterior sampling \cite{CuiLawMarzouk2016}.
Here we investigate how likelihood information can instead guide
the allocation of forward evaluations among candidates drawn from
the prior.
When most posterior mass lies in a region of small prior probability,
evaluating every candidate can devote substantial computation to
regions with little posterior relevance.
We therefore ask whether likelihood evaluations on a finite pilot
sample can be reused to rank a much larger prior sample, allowing
further forward evaluations to focus on the selected candidates.

Our approach encodes the Bayesian update through likelihood-weighted moment geometry. Multiplying the prior by the likelihood changes its polynomial moments; these moments form a weighted Gram matrix, with multivariate Hankel structure in monomial coordinates. In a prior-orthonormal basis, we derive two inexpensive screening scores from this matrix: a Christoffel ratio based on inverse-moment geometry and a Nevai score that averages the likelihood through a polynomial kernel. The pilot evaluations are reused across the candidate pool, so scoring requires no additional forward solves. Polynomial dependence can represent disconnected posterior-relevant regions even when the underlying feature map is linear. Screening alone, however, does not produce exact posterior samples: discarding candidates changes the sampling law unless a separate, properly corrected inference procedure is used.

Christoffel functions and Nevai operators are classical tools of orthogonal-polynomial and moment theory, with established density, support, and localization interpretations \cite{DunklXu2001,KrooLubinsky2013,Lubinsky2011}. Empirical Christoffel functions have been used for support inference and anomaly detection \cite{LasserreSlot2026,LasserrePauwels2019,PauwelsPutinarLasserre2021}, and an infinite-dimensional construction with separate algebraic and harmonic resolutions has been developed for support and trajectory-outlier detection \cite{HenrionLasserre2025}. Weighted Gram matrices also arise in measure modification, where their Cholesky factors connect polynomial families \cite{GumerovRiggSlevinsky2026}. Here we compare the prior with the likelihood-weighted prior to measure posterior relevance. Our contribution concerns the resulting screening scores and their approximation and sampling errors, rather than the existence of the polynomial change of basis.

The analysis connects finite-pilot matrix concentration, quantitative polynomial localization, the discrepancy between the two scores, and posterior-mass regret at a fixed screening budget. For function-valued unknowns, nested feature maps separate three sources of error: unresolved features, finite polynomial resolution, and empirical moment estimation. The conditional likelihood describes the exact Bayesian update visible through the selected features, preserving the evidence and the posterior on feature-measurable events. The weighted moment matrix is a finite-dimensional compression of likelihood multiplication. For Gaussian priors, a separate filtered Nevai construction controls feature selection, localization, and finite-degree truncation; the corresponding rates are not asserted for the unfiltered Hermite kernel or the Gaussian Christoffel ratio.

The numerical experiments examine whether moment geometry provides useful screening at a fixed likelihood budget. In a nonlinear function-coefficient PDE, we compare the scores with direct likelihood regression and forward-response surrogates trained on the same solves. The former use scalar likelihoods, whereas the latter also use the vector of model outputs. The Christoffel ratio improves on the tested scalar regressions but does not outperform the forward surrogates; we therefore report screening quality, information requirements, and computational costs separately. Multimodal and diffusion examples illustrate posterior-relevant geometry, while Gaussian experiments distinguish feature resolution, score calibration, filter truncation, and independent-pilot feature selection.

Section~\ref{sec:hankel} develops likelihood-weighted moment geometry and pilot estimators. Section~\ref{sec:function-space} introduces conditional likelihoods, the operator interpretation, and the three-scale error decomposition. Section~\ref{sec:scores} analyzes the two scores, localization, screening regret, and the filtered Gaussian extension. Section~\ref{sec:numerics} presents the experiments, and Section~\ref{sec:conclusion} discusses the scope and limitations. Appendix~\ref{app:proofs} contains selected proofs.

\section{Posterior-weighted Hankel geometry}
\label{sec:hankel}

Let $\X$ be a separable Banach space, possibly finite-dimensional, equipped with a Borel prior probability measure $\mu_0$. Let $L(\cdot; y): \X\to[0,1]$ be measurable. Define the weighted measure $\nu^y$ and posterior $\mu^y$ by
\begin{equation}
 \frac{\dd\nu^y}{\dd\mu_0}(q)=L(q; y),
 \qquad Z(y):=\int_\X L(q; y)\,\dd\mu_0(q)>0,
 \qquad \frac{\dd\mu^y}{\dd\mu_0}(q)=\frac{L(q; y)}{Z(y)}.
 \label{eq:posterior-measures}
\end{equation}
We first describe polynomial moment geometry for finite-dimensional parameters and then extend it to features of general $q\in\X$. Compact support is imposed only for the bounded-feature concentration estimate below.

A likelihood bounded by a known positive constant can be reduced to this form by division; this scaling leaves the posterior unchanged. The data $y$ are fixed throughout. We retain the superscript $y$ on posterior measures but suppress it from moment and Gram matrices. We write $L$ and $Z$ when no ambiguity arises. Subscripts $d$ and $m$ indicate polynomial degree and feature dimension, respectively, not matrix order. Throughout, $\preceq$ denotes the Loewner order for symmetric matrices. For vectors, $\|\cdot\|_2$ is the Euclidean norm; $\|\cdot\|_{\op}$ denotes the induced operator norm and $\|\cdot\|_{\mathrm F}$ the Frobenius matrix norm. Function norms are written $\|\cdot\|_{L^r(\mu)}$, with $\|\cdot\|_r$ used locally when the reference measure is specified. For probability measures, $\|\mu-\nu\|_{\TV}:=\sup_A|\mu(A)-\nu(A)|$. We denote the one-dimensional standard Gaussian measure by $\gamma_{\mathrm G}:=\mathcal N(0,1)$ and, for a finite coordinate set $S$, its product measure on $\mathbb R^S$ by $\gamma_{\mathrm G,S}:=\bigotimes_{j\in S}\gamma_{\mathrm G}$. Expectations of empirical errors are over the pilot draws. All bases and feature maps are fixed independently of those draws unless a separate selection pilot is explicitly introduced.

\subsection{Monomial moments and the Hankel structure}

In this subsection and the next, take $\X=\R^p$. Assume that the prior has finite moments of every order used below; compact support is sufficient but is not needed for the moment identities. We first make the multi-index notation explicit.  Let $\N_0=\{0, 1, \ldots\}$ and, for $\alpha=(\alpha_1,\ldots,\alpha_p)\in\N_0^p$, define $|\alpha|=\alpha_1+\cdots+\alpha_p$, $q^\alpha=\prod_{j=1}^p q_j^{\alpha_j}$.  For polynomial degree $d$, let $\mathcal A_{p, d}:=\{\alpha\in\N_0^p:|\alpha|\le d\}$ and $s_d:=|\mathcal A_{p, d}|=\binom{p+d}{d}$.
 Define the polynomial space by
\begin{equation}
  \mathcal P_d:=\operatorname{span}\{q^\alpha:\alpha\in\mathcal A_{p, d}\}.
  \label{eq:polynomial-space}
\end{equation}
The function-space construction below uses the corresponding polynomial
space on a feature domain and pulls it back to the original parameter space.
Enumerate $\mathcal A_{p, d}$ as $\alpha^{(0)},\ldots,\alpha^{(s_d-1)}$, with $\alpha^{(0)}=0$. The monomial vector is $\boldsymbol{v}_d(q):=(q^{\alpha^{(j)}})_{j=0}^{s_d-1}\in\R^{s_d}$. For a finite measure $\rho$ on $\X$ with finite moments through degree $2d$, define its degree-$d$ moment matrix by
\begin{equation}
  \mathsf M_d(\rho):=\int_\X \boldsymbol{v}_d(q)\boldsymbol{v}_d(q)^\top\dd\rho(q).
  \label{eq:hankel-monomial}
\end{equation}
For integer row and column indices $i, j\in\{0,\ldots, s_d-1\}$,
$$[\mathsf M_d(\rho)]_{ij} =\int_\X q^{\alpha^{(i)}+\alpha^{(j)}}\dd\rho(q).$$
Here $\alpha^{(i)}$ and $\alpha^{(j)}$ are the multi-indices associated with rows $i$ and $j$, and their addition is componentwise. Thus an entry depends on the associated multi-indices only through their sum, which is the usual multivariate Hankel structure.

For analysis and computation we replace the monomials by a basis that is orthonormal under the prior.  We assume that the prior moment matrix is nondegenerate on the chosen polynomial space, $\mathsf M_d(\mu_0)\succ0$, which means that no polynomial with a nonzero ambient monomial coefficient vector vanishes $\mu_0$-almost everywhere.  Under this assumption there is a nonsingular matrix $\mathsf A_d$ such that $\boldsymbol{\phi}_d(q)=\mathsf A_d \boldsymbol{v}_d(q)$, and
\begin{equation}
  \int_\X\boldsymbol{\phi}_d(q)\boldsymbol{\phi}_d(q)^\top\dd\mu_0(q)=I_{s_d}.
  \label{eq:orthonormal-basis}
\end{equation}
We choose scalar prior-orthonormal basis functions $\{\phi_j\}_{j=0}^{s_d-1}$ with $\phi_0\equiv1$ and collect them into $\boldsymbol{\phi}_d=(\phi_0,\ldots,\phi_{s_d-1})^\top$. Bold symbols distinguish basis and coefficient vectors from their scalar components; the subscript $d$ on $\boldsymbol{\phi}_d$ denotes polynomial resolution, whereas $j$ on $\phi_j$ denotes a component index. Gram matrices are indexed from $0$ throughout. The symbol $\mathsf A_d$ is used only for this change of polynomial basis.  The posterior-weighted moment matrix in prior-orthonormal coordinates is
\begin{equation}
  H_d:=\int_\X\boldsymbol{\phi}_d(q)\boldsymbol{\phi}_d(q)^\top L(q; y)\dd\mu_0(q).
  \label{eq:Hy}
\end{equation}
It is a congruence transform of the monomial Hankel matrix $\mathsf M_d(\nu^y)$ and therefore contains the same moment information, while the prior matrix has been normalized to the identity.  Literal Hankel indexing is present in monomial coordinates; below we retain the term \emph{posterior Hankel matrix} for this prior-orthonormal representation because it is the same moment matrix expressed in a numerically preferable basis.

\begin{proposition}
\label{prop:order}
Assume that the constant polynomial is the first element of \(\boldsymbol{\phi}_d\).  Then
$0\preceq H_d\preceq I_{s_d}$,  $(H_d)_{00}=Z(y)$.
If \(L(q; y)>0\) for \(\mu_0\)-almost every \(q\) and no nonzero polynomial in the chosen space vanishes \(\mu_0\)-almost everywhere, then \(H_d\succ0\).
\end{proposition}

\begin{proof}
For every \(\mathbf{a}\in\R^{s_d}\), $$\mathbf{a}^\top H_d\mathbf{a} =\int L(q; y)[\mathbf{a}^\top\boldsymbol{\phi}_d(q)]^2\dd\mu_0(q)$$ is nonnegative and, since \(L\le1\), is at most \(\|\mathbf{a}\|_2^2\) by \eqref{eq:orthonormal-basis}.  The \((0,0)\) entry is the integral of \(L\) because the first basis function is one.  Strict positivity follows from the stated nondegeneracy.
\end{proof}

\subsection{Likelihood-weighted pilot estimators}

Throughout the paper the empirical moment matrix is formed from ordinary likelihood evaluations at iid prior draws.  Let $Q_i\stackrel{\rm iid}{\sim}\mu_0$,  $i=1,\ldots, M$, and define
\begin{equation}
  \widehat H_d^{(M)}:=\frac1M\sum_{i=1}^M
  L(Q_i; y)\boldsymbol{\phi}_d(Q_i)\boldsymbol{\phi}_d(Q_i)^\top.
  \label{eq:Hhat}
\end{equation}
Then $\E\widehat H_d^{(M)}=H_d$, so the empirical matrix is the standard Monte Carlo estimator of the likelihood-weighted moment matrix.  In particular, its $(0, 0)$ entry is an unbiased estimator of the evidence $Z(y)$. For the concentration bound below, assume additionally that $\operatorname{supp}\mu_0$ is compact and define $\kappa_d:=\sup_{q\in\operatorname{supp}\mu_0}\|\boldsymbol\phi_d(q)\|_2^2\in[1,\infty)$.

\begin{theorem}[Matrix concentration for a weighted pilot]
\label{thm:matrix-concentration}
For every \(t>0\),
\begin{equation}
 \Pp\!\left(\|\widehat H_d^{(M)}-H_d\|_{\op}\ge t \right) \le 2s_d
 \exp\!\left\{-\frac{Mt^2}{2\kappa_d+\frac43 \kappa_d t} \right\}.
 \label{eq:matrix-concentration}
\end{equation}
\end{theorem}

\noindent\emph{The proof is given in Appendix~\ref{app:matrix-concentration}.}

At fixed degree $d$, \eqref{eq:matrix-concentration} implies
\[
 \|\widehat H_d^{(M)}-H_d\|_{\op}
 \le
 \sqrt{\frac{2\kappa_d\log(2s_d/\delta)}{M}}
 +\frac{4\kappa_d\log(2s_d/\delta)}{3M}
\]
with probability at least $1-\delta$, for $0<\delta<1$.
As shown in Section~\ref{sec:scores}, this matrix error directly
bounds the error in the Nevai score. For the Christoffel ratio,
the perturbation bound also depends on $\lambda_{\min}(H_d)$,
reflecting the sensitivity of matrix inversion.

\section{Function-space posterior relevance}
\label{sec:function-space}
Let $\mu_0$ be the prior on the parameter space $\X$ introduced
in Section~\ref{sec:hankel}.
To describe finite feature information, let
$\Xi_m:\X\to\mathbb R^m$ be measurable and write
$\Xi_m(q)=(\xi_1(q),\ldots,\xi_m(q))^\top$.
Define $\mathcal F_m:=\sigma(\Xi_m)$ and assume
$\mathcal F_m\subseteq\mathcal F_{m+1}$.
Thus $\mathcal F_m$ consists of the events determined by the
selected features.
For function-valued parameters, typical choices are the first
$m$ Karhunen--Lo\`eve (KL), Fourier, or wavelet coefficients.
The feature-resolved likelihood is
\begin{equation}
  L_m(q):=\E_{\mu_0}[L(Q)\mid\mathcal F_m](q).
  \label{eq:Lm}
\end{equation}
Here $Q$ denotes a prior-distributed random parameter. The $\sigma$-algebra $\mathcal F_m$ is on $\X$. By the Doob--Dynkin lemma there is a measurable $\bar L_m:\mathbb R^m\to[0,1]$ such that $L_m=\bar L_m\circ\Xi_m$ almost surely. Thus $L_m$ is the orthogonal projection of $L$ onto $L^2(\mathcal X,\mathcal F_m,\mu_0)$. For KL features, this averages over the conditional prior of the omitted coefficients, rather than setting them to zero.

\begin{theorem}[Exact feature posterior and function-space limit]
\label{thm:feature-posterior}
Let $\rho_m^0=(\Xi_m)_\#\mu_0$ and $\rho_m^y=(\Xi_m)_\#\mu^y$, where $\Xi_\#\mu$ denotes the pushforward of a measure $\mu$ under $\Xi$. Then
\begin{equation}
  \frac{\dd\rho_m^y}{\dd\rho_m^0}(\bzeta)=\frac{\bar L_m(\bzeta)}{Z}
  \qquad \rho_m^0\text{-a.e.}
  \label{eq:feature-rn}
\end{equation}
and $\int_\X L_m(q)\,\dd\mu_0(q)=Z$. Define the lifted feature posterior $\dd\mu_m^y(q)=Z^{-1}L_m(q)\,\dd\mu_0(q)$. Then
\begin{equation}
  \mu_m^y(A)=\mu^y(A)
  \qquad\text{for every }A\in\mathcal F_m.
  \label{eq:feature-exact}
\end{equation}
Let $\mathcal F_\infty
:=\sigma\!\left(\bigcup_{m\ge1}\mathcal F_m\right)$.
If every Borel set in $\X$ agrees $\mu_0$-almost surely
with a set in $\mathcal F_\infty$, then
$L_m\to L$ $\mu_0$-almost surely and, for every
$1\le r<\infty$,
\begin{equation}
\begin{aligned}
& \|L_m-L\|_{L^r(\mu_0)}\longrightarrow0,\\
& \|\mu_m^y-\mu^y\|_{\TV}=\frac{1}{2Z}\|L_m-L\|_{L^1(\mu_0)}
 \longrightarrow0.
\end{aligned}
\label{eq:function-TV}
\end{equation}

\end{theorem}

\begin{proof}
For bounded measurable $g$ on the feature space, 
\[
\int g(\Xi_m(q))L(q)\,\dd\mu_0(q) =\int g(\Xi_m(q))L_m(q)\,\dd\mu_0(q),
\]
 which proves \eqref{eq:feature-rn} after pushforward. Taking $g\equiv1$ establishes preservation of $Z$. If $A\in\mathcal F_m$, the defining property of conditional expectation implies $$\int_A L_m(q)\,\dd\mu_0(q)=\int_A L(q)\,\dd\mu_0(q),$$
  proving \eqref{eq:feature-exact}. Finally, $(L_m,\mathcal F_m)$ is a bounded martingale. L\'evy's upward theorem ensures convergence to $\E[L\mid\mathcal F_\infty]=L$ almost surely and in $L^1(\mu_0)$ under the stated completeness assumption. Boundedness upgrades this to every finite $L^r(\mu_0)$. Formula \eqref{eq:function-TV} follows from the common normalization $Z$.
\end{proof}

The theorem shows that feature reduction is not the same as replacing the unknown by a truncated parameter. At fixed $m$, $\mu_m^y$ reproduces the exact posterior on all questions measurable through the chosen features. Increasing $m$ reveals progressively finer posterior information while the mathematical target remains a measure on $\X$.

\subsection{Likelihood multiplication operator and finite-dimensional compression}

Let $\mathcal M_L: L^2(\mu_0)\to L^2(\mu_0)$ be multiplication by the likelihood, $$(\mathcal M_L f)(q)=L(q)f(q).$$  Since $0\le L\le1$, the multiplication operator is self-adjoint and
$0\preceq\mathcal M_L\preceq I$. For the feature levels and degrees
considered below, assume the feature coordinates have finite moments through
order $2d$. For $m\ge1$ and $d\ge0$, let $$\mathcal P_d^{(m)}:=\operatorname{span}\{\bzeta^\alpha:\alpha\in\mathcal A_{m, d}\}$$ be the degree-$d$ polynomial space on $\R^m$, in direct analogy with
\eqref{eq:polynomial-space}. Its pullback through the feature map is the
finite-dimensional space of cylinder polynomials
\begin{equation}
 V_{d,m}:=\left\{r\circ\Xi_m:r\in\mathcal P_d^{(m)}\right\}
 \subset L^2(\mu_0).
 \label{eq:cylinder-space}
\end{equation}
In particular, if $\X=\R^p$ and $\Xi_p$ is the identity, then $V_{d, p}=\mathcal P_d$ as subspaces of $L^2(\mu_0)$, with almost-sure identification understood. For a function-valued $q$,
$\Xi_m(q)$ may consist of its first $m$ KL coefficients, but
$r\circ\Xi_m$ remains a function of the original $q$: pilot weights use
$L(q)$, without replacing the Bayesian model by a truncated-parameter
likelihood. Elements that agree $\mu_0$-almost everywhere are identified. Write $$s_{d, m}:=\dim V_{d, m}\le\binom{m+d}{d},$$ with equality when the degree-$d$ feature moment matrix is nondegenerate. Collect a prior-orthonormal basis into the vector
$\boldsymbol{\phi}_{d, m}=(\phi_0,\ldots,\phi_{s_{d, m}-1})^\top$ with
$\phi_0\equiv1$, and let $P_{d, m}$ be the orthogonal projection onto
$V_{d, m}$. The scalar basis functions $\phi_j$ depend on the fixed resolution $(d, m)$; this dependence is suppressed in their component notation. Thus $H_{d, m}$ below is an $s_{d, m}\times s_{d, m}$ matrix. Choose polynomial representatives $p_j$ on the feature space such that $\phi_j=p_j\circ\Xi_m$ and set
\begin{equation}
 \boldsymbol p_{d, m}(\bxi):=(p_0(\bxi),\ldots, p_{s_{d, m}-1}(\bxi))^\top,
 \qquad \boldsymbol\phi_{d, m}(q)=\boldsymbol p_{d, m}(\Xi_m(q)),
 \label{eq:feature-basis-lift}
\end{equation}
with $p_0\equiv1$. The vector $\boldsymbol p_{d, m}$ is orthonormal under $\rho_m^0$, whereas $\boldsymbol\phi_{d, m}$ is orthonormal under $\mu_0$. If the feature moment matrix is singular, evaluations of polynomial equivalence classes are only intrinsic $\rho_m^0$-almost everywhere; all probabilistic statements use this convention. When it is positive definite, the polynomials define unambiguous pointwise scores on $\mathbb R^m$.

\begin{proposition}
\label{prop:finite-section}
The matrix
\begin{equation}
  H_{d,m}
  =\int_\X \boldsymbol{\phi}_{d, m}(q)\boldsymbol{\phi}_{d, m}(q)^\top L(q)\,\dd\mu_0(q)
  \label{eq:function-H}
\end{equation}
is the matrix representation of the finite section, or finite-dimensional compression, $P_{d, m}\mathcal M_LP_{d, m}$ on $V_{d, m}$. Moreover,
\begin{equation}
  H_{d,m}
  =\int_\X \boldsymbol{\phi}_{d,m}(q)\boldsymbol{\phi}_{d,m}(q)^\top L_m(q)\,\dd\mu_0(q),
  \label{eq:H-Lm}
\end{equation}
so the moment geometry at resolution $m$ depends on the full likelihood only through the exact conditional likelihood $L_m$. 
Let \((d_n, m_n)_{n\ge1}\) be a sequence such that
\(V_{d_n, m_n}\subseteq V_{d_{n+1},m_{n+1}}\) and
\(\overline{\bigcup_{n\ge1}V_{d_n, m_n}}=L^2(\mu_0)\).
Writing \(P_n:=P_{d_n, m_n}\), we then have
$$P_n\mathcal M_LP_n f\to\mathcal M_Lf$$ in $L^2(\mu_0)$ for every $f\in L^2(\mu_0)$.
\end{proposition}

\begin{proof}
The first statement is the definition of the matrix of a compressed multiplication operator. Every product $\phi_j\phi_k$ is $\mathcal F_m$-measurable and integrable by Cauchy--Schwarz. Since $L$ is bounded, conditioning on $\mathcal F_m$ proves \eqref{eq:H-Lm}. For the last claim, the nestedness and density assumptions imply
\(P_n g\to g\) in \(L^2(\mu_0)\) for every
\(g\in L^2(\mu_0)\).
Since \(P_n\) is an orthogonal projection and
\(\mathcal M_L\) is bounded, we have
\[
 \|P_n\mathcal M_LP_n f-\mathcal M_L f\|_2 \le
 \|\mathcal M_L\|\,\|P_n f-f\|_2
 +\|(P_n-I)\mathcal M_L f\|_2
 \longrightarrow 0.
\]

\end{proof}

Equivalently, if $\mathbf{a}$ is the coefficient vector of $P_{d, m}f$ in the chosen basis, then $H_{d, m}\mathbf{a}$ is the coefficient vector of $P_{d, m}\mathcal M_LP_{d, m}f$. The matrix thus represents likelihood reweighting within the selected polynomial space. The density assumption in Proposition~\ref{prop:finite-section} is additional: completeness of the feature filtration alone does not imply polynomial density for an arbitrary prior.

\subsection{Projection benchmark and a three-scale error decomposition}

The operator interpretation suggests a natural projection
benchmark for the screening scores. Since \(\mathcal M_L1=L\),
applying the compressed operator to the constant function
recovers the orthogonal projection of the likelihood onto \(V_{d, m}\).

To express this projection in coordinates, index the chosen
prior-orthonormal basis by \(j=0,\ldots, s_{d, m}-1\), with
\(\phi_0\equiv1\). The entries in column \(0\) of \(H_{d, m}\)
then satisfy $(H_{d, m})_{j0}=\langle L,\phi_j\rangle_{L^2(\mu_0)}$. Define
\[
 G_{d, m}(q) :=\sum_{j=0}^{s_{d, m}-1} (H_{d, m})_{j0}\,\phi_j(q).
\]
Then
\begin{equation}
  G_{d,m}=P_{d,m}L=P_{d,m}L_m.
  \label{eq:G-projection}
\end{equation}
This identity separates feature resolution, controlled by \(m\),
from polynomial approximation within the selected features,
controlled by \(d\). Estimating the projection from a finite
pilot sample introduces a third source of error.

For function-valued draws $Q_i\sim\mu_0$ on $\X$, use the same likelihood-weighted pilot construction and set
\begin{equation*}
\widehat c_j
  =\frac1M\sum_{i=1}^ML(Q_i)\phi_j(Q_i),
  \qquad
  \widehat G_{d,m}^{(M)}=\sum_{j=0}^{s_{d,m}-1}\widehat c_j\phi_j.
\end{equation*}
Write $c_j=(H_{d,m})_{j0}$ for the population coefficients. The empirical polynomial approximation need not take values
in \([0,1]\), even though the likelihood does. We therefore
define $$\Pi_{[0,1]}(t):=\min\left\{1,\max\{0,t\}\right\}$$ and its pointwise clipped version $$\widetilde G_{d, m}^{(M)}:=\Pi_{[0,1]}\widehat G_{d, m}^{(M)}.$$  Since \(L(q)\in[0,1]\), this clipping cannot increase the
pointwise error relative to \(L\). The clipped function remains
\(\mathcal F_m\)-measurable, although it need not belong to
\(V_{d,m}\).

\begin{theorem}[Function-space Hankel error decomposition]
\label{thm:function-error}
With iid prior draws and deterministic likelihood weights,

\[
 \E\|\widehat G_{d, m}^{(M)}-G_{d, m}\|_{L^2(\mu_0)}^2 \le \frac{s_{d,m}}{M}
.
\]
Consequently,
\begin{equation}
  \E\|\widetilde G_{d, m}^{(M)}-L\|_{L^2(\mu_0)}
  \le
  \underbrace{\|L-L_m\|_2}_{\text{feature error}}
  +\underbrace{\|L_m-P_{d,m}L_m\|_2}_{\text{polynomial error}}
  +\underbrace{\sqrt{s_{d,m}/M}}_{\text{pilot error}}.
  \label{eq:three-error}
\end{equation}
In particular, along any schedule $(M_n, d_n, m_n)$ for which the three terms on the right vanish, $\widetilde G_{d_n, m_n}^{(M_n)}\to L$ in mean $L^2(\mu_0)$.
\end{theorem}

\begin{proof}
Since $0\le L\le1$ and the basis is prior-orthonormal,

\[
 \operatorname{Var}_{\mu_0}(L\phi_j)\le\E[L^2\phi_j^2]
\le\E[L\phi_j^2]\le\E[\phi_j^2]=1
.
\]
By independence of the pilot draws and orthonormality,
\[
 \E\|\widehat G_{d, m}^{(M)}-G_{d, m}\|_2^2
 =\sum_j\E|\widehat c_j-c_j|^2
 =M^{-1}\sum_j\operatorname{Var}_{\mu_0}(L\phi_j)
 \le s_{d, m}/M.
\]
Clipping is nonexpansive relative to $L\in[0,1]$; the triangle
inequality, \eqref{eq:G-projection}, and Jensen's inequality yield
\eqref{eq:three-error} and the asserted convergence.
\end{proof}

Equation~\eqref{eq:three-error} is a direct projection benchmark.  It separates three quantities that are often conflated in high-dimensional inversion: how much posterior information lies outside the selected function modes, how well the conditional likelihood is represented inside the selected polynomial space, and how accurately the required moments are estimated.  The screening procedures below use Nevai and Christoffel scores; the first-column projection is a comparison method, and Section~\ref{sec:scores} establishes the corresponding three-scale Nevai bound.

The compact-domain concentration result in Section~\ref{sec:hankel} is useful for bounded features, but Gaussian function-space priors lead to unbounded Hermite features.  The following second-moment estimate covers that setting without a uniform envelope.

\begin{proposition}
\label{prop:unbounded-feature}
Let the components of $\boldsymbol{\phi}_{d, m}$ be orthonormal in $L^2(\mu_0)$ and suppose $\chi_{d, m}:=\E_{\mu_0}\|\boldsymbol{\phi}_{d, m}(Q)\|_2^4<\infty$. With iid prior draws $Q_i\sim\mu_0$ on $\X$, define
\begin{equation}
  \widehat H_{d, m}^{(M)}
  :=\frac1M\sum_{i=1}^M
  L(Q_i)\boldsymbol{\phi}_{d, m}(Q_i)\boldsymbol{\phi}_{d, m}(Q_i)^\top.
  \label{eq:function-Hhat}
\end{equation}
Then
\begin{equation}
\label{eq:moment-error-bounds}
 \E\|\widehat H_{d, m}^{(M)}-H_{d, m}\|_{\mathrm F}^2\le\chi_{d, m}/M,
\quad \E\|\widehat H_{d, m}^{(M)}-H_{d,m}\|_{\op}\le\sqrt{\chi_{d, m}/M}.
\end{equation}
\end{proposition}

\begin{proof}
Write $$X_i=L(Q_i)\boldsymbol{\Phi}_i\boldsymbol{\Phi}_i^\top$$ with $\boldsymbol{\Phi}_i=\boldsymbol{\phi}_{d, m}(Q_i)$.  By independence and centering, $$\E\bigl\|\widehat H_{d,m}^{(M)}-H_{d, m}\bigr\|_{\mathrm F}^2 =\frac1M\E\|X_1-H_{d, m}\|_{\mathrm F}^2 \le \frac1M\E\|X_1\|_{\mathrm F}^2.$$ Since $$\|\boldsymbol{\phi}\boldsymbol{\phi}^\top\|_{\mathrm F}^2=\|\boldsymbol{\phi}\|_2^4$$ and $0\le L\le1$, $$\E\|X_1\|_{\mathrm F}^2 =\E\!\left[L(Q_1)^2\|\boldsymbol{\phi}_{d,m}(Q_1)\|_2^4\right] \le \chi_{d,m}.$$ This proves the first estimate in \eqref{eq:moment-error-bounds}; the second follows from $\|A\|_{\op}\le\|A\|_{\mathrm F}$ and Jensen's inequality.
\end{proof}
This estimate applies, in particular, to Gaussian priors with
Hermite cylinder features. For each fixed \((d, m)\), the basis
contains finitely many polynomials of Gaussian coordinates,
so \(\chi_{d, m}<\infty\). Hence the empirical moment matrix
satisfies $$\E\|\widehat H_{d, m}^{(M)}-H_{d, m}\|_{\mathrm F}^2\to0$$ as
\(M\to\infty\), without a compact-support assumption.
If \(d\) and \(m\) also increase with \(M\), the same estimate
ensures convergence provided that \(\chi_{d, m}/M\to0\).

\subsection{Posterior-relevance spectrum}

Assume in this subsection the completeness condition of
Theorem~\ref{thm:feature-posterior}, so that $L_m\to L$ in $L^2(\mu_0)$.
Let $\mathcal F_0$ be the trivial $\sigma$-field, so $L_0=Z$. The martingale increments $$D_m:=L_m-L_{m-1}$$ are orthogonal in $L^2(\mu_0)$. Define the relevance energy $$\mathcal E_m:=\|D_m\|_2^2.$$  Then
\begin{align*}
\begin{aligned}
&\|L_m-Z\|_2^2=\sum_{k=1}^m\mathcal E_k,\\  
&\|L-L_m\|_2^2=\sum_{k>m}\mathcal E_k.
\end{aligned}
\end{align*}
When $\operatorname{Var}_{\mu_0}(L)>0$, the resolved fraction
\begin{equation}
  \mathfrak R_m:=1-\frac{\|L-L_m\|_2^2}{\operatorname{Var}_{\mu_0}(L)}
  =\frac{\operatorname{Var}_{\mu_0}(L_m)}{\operatorname{Var}_{\mu_0}(L)}
  \label{eq:resolved-fraction}
\end{equation}
is nondecreasing and converges to one.  The sequence $(\mathcal E_m)$ is relative to the chosen nested ordering of the features: reordering the coordinates generally changes the individual increments, although the cumulative resolved fraction at a fixed generated $\sigma$-algebra is unchanged.  Within a fixed ordering the spectrum is posterior-driven: a mode with small prior variance may still carry substantial relevance energy if the data are sensitive to it.

For a Gaussian prior on a separable Hilbert space with a trace-class covariance and its KL expansion 
\begin{align*}
q=\bar q+\sum_{k\ge1}\sqrt{\lambda_k}\,\xi_ke_k,  \qquad \xi_k\stackrel{\rm iid}{\sim}\gamma_{\mathrm G},
\end{align*}
we take $\Xi_m(q)=(\xi_1(q),\ldots,\xi_m(q))^\top$ and Hermite cylinder polynomials. The union of finite Hermite chaoses is dense in $L^2(\mu_0)$, so the preceding theory applies without replacing the posterior target by a fixed truncation. This differs from likelihood-informed subspace methods, which identify informative linear directions through gradient or Hessian information \cite{CuiLawMarzouk2016}; here posterior relevance is recovered from likelihood-weighted moments.

Henrion and Lasserre's infinite-dimensional Christoffel construction \cite{HenrionLasserre2025} provides a complementary direct Hilbert-space viewpoint with algebraic and harmonic degrees. Our indices $(d, m)$ play analogous resolution roles, but our object is likelihood-weighted prior geometry and our target is Bayesian relevance rather than support detection.

\section{Christoffel and Nevai screening}
\label{sec:scores}

\subsection{Score construction and stability}

We use the feature formulation of Section~\ref{sec:function-space}: the unknown is $q\in\X$, while scores are evaluated at $\bxi=\Xi_m(q)\in\mathbb R^m$. The finite-dimensional parameter model is recovered with $m=p$ and the identity feature map. Define the weighted feature measure $\nu_m:=\bar L_m\rho_m^0=(\Xi_m)_\#\nu^y$. The same moment matrix can be written as
\begin{equation}
\begin{aligned}
 H_{d, m}
 &=\int_\X L(q)\boldsymbol\phi_{d, m}(q)\boldsymbol\phi_{d, m}(q)^\top\,\dd\mu_0(q)\\
 &=\int_{\mathbb R^m}\bar L_m(\bxi)\boldsymbol p_{d, m}(\bxi)\boldsymbol p_{d, m}(\bxi)^\top\,\dd\rho_m^0(\bxi).
\end{aligned}
\label{eq:feature-matrix-bridge}
\end{equation}
Thus the pilot still uses the original likelihood; the finite feature map restricts the scoring functions, not the underlying parameter space.

Define the Christoffel--Darboux kernel and its diagonal by
\begin{equation}
 K_{d,m}(\bxi, \bzeta):=\boldsymbol p_{d, m}(\bxi)^\top\boldsymbol p_{d, m}(\bzeta),
 \qquad \mathcal K_{d, m}(\bxi):=K_{d, m}(\bxi, \bxi).
 \label{eq:kernel}
\end{equation}
For the selected polynomial span $W_{d, m}:=\operatorname{span}\{p_j: 0\le j<s_{d, m}\}$ and a measure $\rho$ with positive definite Gram matrix on this span, write
\[
 \Lambda_{d,m}^{\rho}(\bxi):=\min_{\substack{r\in W_{d,m}\\r(\bxi)=1}}
 \int_{\mathbb R^m}r(\bzeta)^2\,\dd\rho(\bzeta).
\]
When the degree-$d$ feature moment matrix is nondegenerate, $W_{d, m}=\mathcal P_d^{(m)}$. 
Prior orthonormality shows $\Lambda_{d, m}^{\rho_m^0}(\bxi)=\mathcal K_{d, m}(\bxi)^{-1}$. 
If $H_{d, m}\succ0$, define the Christoffel relevance score by
\begin{equation}
 R_{d,m}(\bxi):=\frac{\Lambda_{d, m}^{\nu_m}(\bxi)}{\Lambda_{d, m}^{\rho_m^0}(\bxi)}
 =\frac{\mathcal K_{d, m}(\bxi)}{\boldsymbol p_{d, m}(\bxi)^\top H_{d, m}^{-1}\boldsymbol p_{d, m}(\bxi)}.
 \label{eq:Rscore}
\end{equation}
The Nevai relevance score uses the same matrix without inversion:
\begin{equation}
 T_{d, m}(\bxi):=\frac{\boldsymbol p_{d, m}(\bxi)^\top H_{d, m}\boldsymbol p_{d, m}(\bxi)}{\mathcal K_{d, m}(\bxi)}.
 \label{eq:Tscore}
\end{equation}
The corresponding scores for a parameter $q$ are $R_{d, m}(\Xi_m(q))$ and $T_{d, m}(\Xi_m(q))$.  In the identity-feature model we abbreviate these as $R_d(q)$ and $T_d(q)$.

\begin{proposition}
\label{prop:scores}
Assume $H_{d, m}\succ0$. Then $0<R_{d, m}(\bxi)\le T_{d, m}(\bxi)\le1$, and
\begin{equation}
 T_{d, m}(\bxi)=\int_{\mathbb R^m}\bar L_m(\bzeta)\,\dd\pi_{d,m,\bxi}(\bzeta),\qquad
 \dd\pi_{d, m, \bxi}(\bzeta):=\frac{K_{d, m}(\bxi,\bzeta)^2}{\mathcal K_{d, m}(\bxi)}\dd\rho_m^0(\bzeta),
\label{eq:nevai-average}
\end{equation}
where $\pi_{d, m,\bxi}$ is a probability measure. The Nevai representation and $0\le T_{d, m}\le1$ remain valid when $H_{d, m}$ is singular.
\end{proposition}
\begin{proof}
Set $$\mathbf u=\boldsymbol p_{d, m}(\bxi)/\sqrt{\mathcal K_{d, m}(\bxi)}.$$ 
Orthonormality and $0\le\bar L_m\le1$ imply $$0\preceq H_{d, m}\preceq I,$$ 
hence $T_{d, m}\le1$. By Cauchy--Schwarz, $$(\mathbf u^\top H_{d, m}\mathbf u)(\mathbf u^\top H_{d, m}^{-1}\mathbf u)\ge1,$$ 
proving the ordering. Expanding the quadratic form proves the integral representation, while orthonormality implies $$\int K_{d, m}(\bxi, \bzeta)^2\dd\rho_m^0(\bzeta)=\mathcal K_{d, m}(\bxi).$$ The latter arguments do not require invertibility.
\end{proof}

\paragraph{Matrix perturbations.}

At fixed $(d, m)$, both scores are normalized quadratic forms: with $$\mathbf u=\boldsymbol p_{d, m}(\bxi)/\sqrt{\mathcal K_{d, m}(\bxi)},$$
 they are $T_{d, m}=\mathbf u^\top H_{d, m}\mathbf u$ and $R_{d, m}=(\mathbf u^\top H_{d, m}^{-1}\mathbf u)^{-1}$. Their perturbation bounds are purely matrix statements and do not depend on the dimension of $\X$.

\begin{proposition}
\label{prop:score-stability}
Let $H\succ0$ and $\widehat H=H+E$ be symmetric.  For a unit vector $\mathbf{u}$, define 
\begin{align*}
T=\mathbf{u}^\top H\mathbf{u},  \quad \widehat T=\mathbf{u}^\top\widehat H\mathbf{u},  \quad R=(\mathbf{u}^\top H^{-1}\mathbf{u})^{-1}.
\end{align*}
 Then
$|\widehat T-T|\le\|E\|_{\op}$.
If $\delta:=\|H^{-1/2}EH^{-1/2}\|_{\op}<1$, then $\widehat H\succ0$ and, with $\widehat R=(\mathbf{u}^\top\widehat H^{-1}\mathbf{u})^{-1}$,
\begin{equation}
  (1-\delta)R\le\widehat R\le(1+\delta)R.
  \label{eq:christoffel-relative}
\end{equation}
Moreover,
\begin{equation}
  \delta\le\frac{\|E\|_{\op}}{\lambda_{\min}(H)}.
  \label{eq:delta-lambda}
\end{equation}
\end{proposition}

\begin{proof}
The first bound follows from $|\mathbf{u}^\top E\mathbf{u}|\le\|E\|_{\op}$.  By the definition of $\delta$, $(1-\delta)H\preceq\widehat H\preceq(1+\delta)H$. Inverting these inequalities reverses the Loewner order: $$(1+\delta)^{-1}H^{-1} \preceq\widehat H^{-1} \preceq(1-\delta)^{-1}H^{-1}.$$ Taking quadratic forms with $\mathbf{u}$ and reciprocals yields \eqref{eq:christoffel-relative}.  Finally, \eqref{eq:delta-lambda} follows from
$$\|H^{-1/2}\|_{\op}^2=\lambda_{\min}(H)^{-1}.$$
\end{proof}

The Nevai score is $1$-Lipschitz in the moment matrix and remains defined
when the empirical matrix is singular. The Christoffel ratio requires
inversion, and its relative perturbation bound depends on
$\lambda_{\min}(H_{d, m})$. Thus localization at higher degree must be balanced
against the accuracy required for empirical Christoffel inversion; neither
score is uniformly more accurate than the other.

The averaging representation also controls localization bias. If $|\bar L_m(\bzeta)-\bar L_m(\bxi)|\le\omega_m(\|\bzeta-\bxi\|_2)$ on the feature support, then
\begin{equation}
 |T_{d, m}(\bxi)-\bar L_m(\bxi)|\le\Delta_{d, m}(\bxi)
 :=\int\omega_m(\|\bzeta-\bxi\|_2)\,\dd\pi_{d,m,\bxi}(\bzeta).
 \label{eq:modulus}
\end{equation}
This follows by subtracting $\bar L_m(\bxi)$ inside the probability average. For an $\ell_m$-Lipschitz feature likelihood, take $\omega_m(r)=\min\{1,\ell_m r\}$. Localization here is in feature coordinates; it does not by itself control the omitted-feature error $L-L_m$.

\subsection{Localization and comparison of the scores}
\label{sec:quantitative-localization}

The preceding bound depends on kernel localization. We now quantify this dependence for the uniform feature prior $\rho_m^0=\Unif([-1,1]^m)$, with $m$ fixed. Here $W_{d,m}=\mathcal P_d^{(m)}$, whose pullback is $V_{d,m}$ in \eqref{eq:cylinder-space}. Its orthonormal basis consists of the products $p_\alpha(\bxi)=\prod_{j=1}^m \mathsf P_{\alpha_j}(\xi_j)$ with $|\alpha|\le d$, where $\mathsf P_k$ is the degree-$k$ orthonormal Legendre polynomial under $\Unif([-1,1])$. Using the probability measure in \eqref{eq:nevai-average}, set
\begin{align*}
\begin{aligned}
 \mathfrak m_{2,d,m}(\bxi)&:=\int_{[-1,1]^m}\|\bzeta-\bxi\|_2^2\,\dd\pi_{d,m,\bxi}(\bzeta),\\
 \mathcal V_{d,m}&:=\int_{[-1,1]^m}\mathfrak m_{2,d,m}(\bxi)\,\dd\rho_m^0(\bxi).
\end{aligned}
\end{align*}

\begin{theorem}[Quantitative localization for total-degree Legendre spaces]
\label{thm:legendre-localization}
For every $m\ge1$ and $d\ge0$,
\begin{equation}
  \mathcal V_{d, m}
  \le
  \frac{2m(2\pi)^m}{3} \frac{\binom{d+m-1}{m-1}}{(\lfloor d/m\rfloor+1)^m}.
  \label{eq:Vdm-exact-bound}
\end{equation}
In particular, if $d\ge m$, then
\begin{equation}
  \mathcal V_{d, m}
  \le \frac{C_m}{d+1},
  \qquad
  C_m:=\frac{2^{m+1}(2\pi)^m m^{m+1}}{3(m-1)!}.
  \label{eq:Vdm-Cm}
\end{equation}
Consequently, if the feature likelihood $\bar L_m$ is $\ell_m$-Lipschitz, then the population Nevai score satisfies
\begin{equation}
  \beta_{d, m}
  :=\|T_{d,m}-\bar L_m\|_{L^2(\rho_m^0)}
  \le
  \ell_m\sqrt{\mathcal V_{d, m}}
  \le
  \ell_m\sqrt{\frac{C_m}{d+1}},
  \qquad d\ge m.
  \label{eq:nevai-legendre-rate}
\end{equation}
Thus, for fixed $m$, the present argument establishes the explicit rate
$\beta_{d,m}=O_m(d^{-1/2})$.
\end{theorem}

\noindent\emph{The proof is given in Appendix~\ref{app:legendre-localization}.}

The constant $C_m$ is intentionally explicit rather than optimized.  Its rapid
growth with $m$ shows that Theorem~\ref{thm:legendre-localization} is primarily
a fixed-feature-dimension result.  The theorem supplies a quantitative degree
rate once $m$ is fixed; it should not be read as a dimension-robust estimate or
as a practical joint complexity bound for growing $m, d, M$.

The same kernel localization also links the two relevance scores quantitatively. The following Schur-complement argument requires at least two basis functions; for the one-dimensional constant-polynomial space the two scores coincide with $Z$ and their discrepancy is zero.

\begin{proposition}
\label{thm:score-discrepancy}
Fix $d, m$ with $s_{d, m}\ge2$ and suppose
\begin{equation}
  0<a_m\le \bar L_m(\bzeta)\le1
  \qquad \rho_m^0\text{-a.e.}
  \label{eq:positive-feature-likelihood}
\end{equation}
Then
\begin{equation}
  0\le T_{d,m}(\bxi)-R_{d,m}(\bxi)
  \le
  \frac{1}{a_m}
  \operatorname{Var}_{\pi_{d,m,\bxi}}(\bar L_m).
  \label{eq:variance-score-gap}
\end{equation}
If $\bar L_m$ is $\ell_m$-Lipschitz, then
\begin{equation}
  0\le T_{d,m}(\bxi)-R_{d,m}(\bxi)
  \le\frac{\ell_m^2}{a_m}\mathfrak m_{2,d,m}(\bxi).
  \label{eq:lipschitz-score-gap}
\end{equation}
\end{proposition}

The block-matrix identity and proof are given in Appendix~\ref{app:score-discrepancy}.

\begin{corollary}
\label{cor:two-score-rate}
Under the assumptions of Theorem~\ref{thm:legendre-localization}
and Proposition~\ref{thm:score-discrepancy}, for $d\ge m$,
\begin{equation}
\label{eq:two-score-rates}
\begin{aligned}
 &\|T_{d,m}-\bar L_m\|_{L^2(\rho_m^0)}
 &&\le \ell_m\sqrt{\frac{C_m}{d+1}},\\
 &\|T_{d,m}-R_{d,m}\|_{L^1(\rho_m^0)}
 &&\le \frac{\ell_m^2C_m}{a_m(d+1)},\\
 &\|R_{d,m}-\bar L_m\|_{L^2(\rho_m^0)}
 &&\le \ell_m\left(1+a_m^{-1/2}\right)
 \sqrt{\frac{C_m}{d+1}}.
\end{aligned}
\end{equation}
Thus both population scores converge to the feature likelihood
in $L^2(\rho_m^0)$ at rate $O_m(d^{-1/2})$, while their
$L^1(\rho_m^0)$ discrepancy is $O_m(d^{-1})$.
\end{corollary}

\begin{proof}
The first estimate follows from \eqref{eq:nevai-legendre-rate}.
Integrating \eqref{eq:lipschitz-score-gap} and applying
\eqref{eq:Vdm-Cm} proves the second.
Since $0\le T_{d,m}-R_{d,m}\le1$, we have
$$\|T_{d,m}-R_{d,m}\|_{L^2(\rho_m^0)}^2
\le \|T_{d,m}-R_{d,m}\|_{L^1(\rho_m^0)}.$$
The third estimate then follows from the first two
and the triangle inequality.
\end{proof}

\paragraph{Finite-pilot consequences.}
Under the assumptions of Corollary~\ref{cor:two-score-rate}, fix $d\ge m$ and let $\varepsilon_{d, m, M} :=\|\widehat H_{d, m}^{(M)}-H_{d, m}\|_{\op}$ and suppose $\varepsilon_{d, m, M}<a_m$.  In particular, this condition
ensures that the empirical matrix is positive definite; it need not be
invertible without it.  Let $\widehat T_{d, m, M}$ and
$\widehat R_{d, m, M}$ be the empirical scores obtained from
$\widehat H_{d, m}^{(M)}$.  Then
\begin{equation}
\label{eq:empirical-score-rates}
\begin{aligned}
 \|\widehat T_{d, m, M}-\bar L_m\|_{L^2(\rho_m^0)}
 &\le
 \ell_m\sqrt{\frac{C_m}{d+1}}
 +\varepsilon_{d, m, M},
 \\
 \|\widehat R_{d, m, M}-\bar L_m\|_{L^2(\rho_m^0)}
 &\le
 \ell_m\left(1+a_m^{-1/2}\right)
 \sqrt{\frac{C_m}{d+1}}
 +\frac{\varepsilon_{d, m, M}}{a_m}.
\end{aligned}
\end{equation}

To see this, the normalized quadratic-form representation yields
$$\|\widehat T_{d, m, M}-T_{d, m}\|_\infty\le\varepsilon_{d, m, M}.$$
For the Christoffel score, Proposition~\ref{prop:score-stability} and
$H_{d, m}\succeq a_mI$ imply $$\|\widehat R_{d, m, M}-R_{d,m}\|_\infty \le \frac{\varepsilon_{d, m, M}}{a_m},$$
 using $0<R_{d,m}\le1$.  Combining these estimates with
Corollary~\ref{cor:two-score-rate} proves the result.

The two rates should not be interpreted as a universal statement that the two scores have identical sharp asymptotics, nor as an ordering of their approximation quality.  These are common sufficient upper bounds.  The sharper $O_m(d^{-1})$ estimate concerns only the discrepancy between the scores in $L^1$.  The lower bound $a_m$ is needed for the inverse-matrix Christoffel score and may be small for a concentrated likelihood; the Nevai localization bound itself does not require it.

In computations we may replace $\widehat H$ by $\widehat H+\tau I$ with a stated ridge $\tau>0$. This defines a regularized score, not the unregularized ratio in \eqref{eq:Rscore}. Relative to the population matrix, its perturbation is $E+\tau I$; the preceding bound applies if $\|E\|_{\op}+\tau<a_m$, with that sum replacing the empirical error. Consistency with the unregularized score requires control of the ridge bias as well as pilot error.

\subsection{Empirical approximation and screening guarantees}

Using the empirical matrix in \eqref{eq:function-Hhat}, define
\begin{equation*}
 \widehat T_{d,m,M}(\bxi):=\frac{\boldsymbol p_{d,m}(\bxi)^\top\widehat H_{d,m}^{(M)}\boldsymbol p_{d,m}(\bxi)}{\mathcal K_{d,m}(\bxi)},
 \qquad \widetilde T_{d,m,M}:=\Pi_{[0,1]}\widehat T_{d,m,M}.
\end{equation*}
Normalizing the evaluation vector, we obtain
\begin{equation}
 \sup_{\bxi\in\mathbb R^m}|\widehat T_{d,m,M}(\bxi)-T_{d,m}(\bxi)|
 \le\|\widehat H_{d,m}^{(M)}-H_{d,m}\|_{\op}.
 \label{eq:T-op}
\end{equation}
This inequality holds for the chosen polynomial representatives; its probabilistic consequences are unchanged by replacing them with representatives equal $\rho_m^0$-almost everywhere. Clipping cannot increase error relative to a target in $[0,1]$. The localization error introduced in \eqref{eq:nevai-legendre-rate} extends to general feature priors through
\[
 \beta_{d,m}=\|T_{d,m}-\bar L_m\|_{L^2(\rho_m^0)}
 =\|T_{d,m}\circ\Xi_m-L_m\|_{L^2(\mu_0)}.
\]

\begin{theorem}[Three-scale consistency of the function-space Nevai score]
\label{thm:function-nevai}
Suppose $\chi_{d,m}<\infty$. Then
\begin{equation}
\begin{split}
  \E\bigl\|\widetilde T_{d, m, M}\circ\Xi_m-L\bigr\|_{L^2(\mu_0)}
  \le\underbrace{\|L-L_m\|_2}_{\text{feature error}}
  +\underbrace{\beta_{d, m}}_{\text{Nevai localization error}}+\underbrace{\sqrt{\chi_{d, m}/M}}_{\text{pilot error}}.
\end{split}
  \label{eq:nevai-three-scale}
\end{equation}
Consequently, along any sequence $(M_n,d_n,m_n)$ for which all three terms on the right-hand side vanish, the empirical Nevai score converges to the full likelihood in mean $L^2(\mu_0)$.
\end{theorem}

\begin{proof}
For every feature value $\bzeta$, normalization by the Christoffel--Darboux diagonal implies 
$$|\widehat T_{d, m, M}(\bzeta)-T_{d, m}(\bzeta)| \le \|\widehat H_{d, m}^{(M)}-H_{d, m}\|_{\op}.$$ Because clipping is nonexpansive relative to the $[0,1]$-valued target $L$, the triangle inequality yields
\[
\begin{split}
  \|\widetilde T_{d, m, M}\circ\Xi_m-L\|_2
  \le\|L-L_m\|_2
  +\|T_{d, m}\circ\Xi_m-L_m\|_2+\|\widehat H_{d, m}^{(M)}-H_{d, m}\|_{\op}.
\end{split}
\]
Take expectations and apply Proposition~\ref{prop:unbounded-feature}.
\end{proof}

Theorem~\ref{thm:function-nevai} is the function-space analogue of the compact-domain bound below.  It makes the three algorithmic resolutions explicit: $m$ controls how much of the function is visible, $d$ controls how well the conditional likelihood is localized by the polynomial kernel, and $M$ controls moment-estimation error.  On compact Legendre feature domains, Theorem~\ref{thm:legendre-localization}
controls $\beta_{d,m}$ explicitly.  No such rate is asserted here for the
unfiltered total-degree Hermite score; Section~\ref{sec:gaussian-adaptive}
introduces a separate filtered Gaussian extension for which a verifiable
finite-polynomial bound is available.

\begin{remark}
\label{thm:nevai-bound}
Consider the finite-dimensional setting $\X=\mathbb R^p$
with compactly supported prior $\mu_0$ and identity feature map
$\Xi_p(q)=q$. We use the abbreviations
$T_d$, $\widehat T_{d,M}$, and $\widetilde T_{d,M}$.
Suppose
$|L(q)-L(q')|\le\omega(\|q-q'\|_2)$
for $q,q'\in\operatorname{supp}\mu_0$, and let $\Delta_d$
be the corresponding localization bound in \eqref{eq:modulus}.
For every $t>0$, with probability at least
$1-2s_d\exp\{-Mt^2/(2\kappa_d+\frac43\kappa_dt)\}$,
\[
 |\widetilde T_{d, M}(q)-L(q)|
 \le \Delta_d(q)+t
 \quad\text{for all } q\in\operatorname{supp}\mu_0.
\]
In particular, along a degree sequence $d_M\to\infty$, if
\[
 \sup_{q\in\operatorname{supp}\mu_0}\Delta_{d_M}(q)\to0,
 \qquad
 \frac{\kappa_{d_M}\log(2s_{d_M})}{M}\to0,
\]
then
\[
 \sup_{q\in\operatorname{supp}\mu_0}
 |\widetilde T_{d_M,M}(q)-L(q)|
 \xrightarrow{\Pp}0.
\]
The probability bound follows from
Theorem~\ref{thm:matrix-concentration},
\eqref{eq:T-op}, and \eqref{eq:modulus}.
Under the stated growth condition, the matrix error converges
to zero in probability, which proves the final assertion.
\end{remark}

\paragraph{From approximation to screening.}
The preceding estimates control likelihood approximation. Screening instead asks how much posterior mass is lost when a score replaces likelihood ranking at a fixed retained prior fraction. The following result connects these questions: the loss in captured mass is controlled by the $L^1$ approximation error.

\begin{theorem}[Posterior-mass regret at a fixed screening budget]
\label{thm:fixed-budget-regret}
Let $\mu_0$ be a nonatomic probability measure, let $L$ be measurable with $0\le L\le1$, let $0<\alpha<1$, and let
$s\in L^1(\mu_0)$ be a real-valued score.  Set
$Z=\int L\,\dd\mu_0>0$.  Let $A_\alpha^*$ maximize
$\int_A L\,\dd\mu_0$ over all measurable sets with $\mu_0(A)=\alpha$,
and let $A_\alpha^s$ maximize $\int_A s\,\dd\mu_0$ under the same
constraint.  The maximizing sets can be taken as upper level sets, with
ties resolved by subsets of the appropriate level sets. Then
\begin{equation}
 0\le \mu^y(A_\alpha^*)-\mu^y(A_\alpha^s)
 \le \frac{\|L-s\|_{L^1(\mu_0)}}{Z}.
 \label{eq:fixed-budget-regret}
\end{equation}
The statement applies pathwise to a pilot-dependent measurable score.
\end{theorem}
\begin{proof}
Set $D=A_\alpha^*\setminus A_\alpha^s$ and
$E=A_\alpha^s\setminus A_\alpha^*$.
The optimality of $A_\alpha^s$ for $s$ implies
$\int_D s\,\dd\mu_0\le\int_E s\,\dd\mu_0$.
Hence 
\[
\begin{aligned}
 Z\{\mu^y(A_\alpha^*)-\mu^y(A_\alpha^s)\}
 &=\int_D L\,\dd\mu_0-\int_E L\,\dd\mu_0\\
 &\le\int_D(L-s)\,\dd\mu_0+\int_E(s-L)\,\dd\mu_0
 \le\|L-s\|_1.
\end{aligned}
\]
 The first inequality in \eqref{eq:fixed-budget-regret} is the optimality
of $A_\alpha^*$ for $L$.
\end{proof}

The same comparison holds for a finite candidate pool when both rankings retain exactly $k$ candidates. With likelihoods $\ell_i$, scores $s_i$, and $\sum_i\ell_i>0$, the difference in captured reference mass is at most $\sum_i|\ell_i-s_i|/\sum_i\ell_i$, by the same set-difference argument. This is an empirical-pool statement; its relation to population mass requires a separate sampling analysis. In either version the factor $1/Z$ explains why a small absolute likelihood error need not guarantee useful screening when the evidence is small.

\paragraph{Legendre screening bound.}
Under the nonatomic-prior assumptions above and the hypotheses of Theorem~\ref{thm:legendre-localization},
let $Q_i$ be iid prior pilot draws and put $\eta_m:=\|L-L_m\|_{L^2(\mu_0)}$. The moment quantity $\chi_{d,m}$ from Proposition~\ref{prop:unbounded-feature} is finite in this compact Legendre setting.
Then the clipped empirical Nevai score satisfies, for $d\ge m$,
\begin{equation*}
\E_{\rm pilot}\!\left[\mu^y(A_\alpha^*)-
 \mu^y(A_\alpha^{\widetilde T})\right]
 \le \frac1Z\left(
 \eta_m+\ell_m\sqrt{\frac{C_m}{d+1}}
 +\sqrt{\frac{\chi_{d,m}}{M}}\right).
\end{equation*}
Here the feature-resolved Legendre model must satisfy the hypotheses of
Theorem~\ref{thm:legendre-localization}; the conclusion does not assign
that localization rate to an unrelated Gaussian/Hermite model.

This follows by applying Theorem~\ref{thm:fixed-budget-regret} pathwise with
$s=\widetilde T_{d, m, M}\circ\Xi_m$, abbreviated as $\widetilde T$, and using
$\|L-\widetilde T\|_1\le\|L-\widetilde T\|_2$, followed by
Theorem~\ref{thm:function-nevai},
Proposition~\ref{prop:unbounded-feature}, and
\eqref{eq:nevai-legendre-rate}; the clipping map is nonexpansive relative
to the $[0,1]$-valued target.

Likelihood approximation is sufficient but not necessary for useful screening. Strictly increasing transformations preserve a score's ranking, so good mass capture can coexist with poor calibration of its values. The factor $1/Z$ also makes the bound sensitive to likelihood concentration.

\paragraph{Screening procedure.}
Choose a feature map and a prior-orthonormal polynomial basis independently of the estimation pilot. Draw an independent pilot from the prior, evaluate its likelihoods, and form the likelihood-weighted moment matrix. For each new candidate evaluate the Nevai quadratic form; the Christoffel ratio may also be used when the empirical matrix is positive definite, or with a stated regularization. Retain a chosen fraction under the resulting ranking and evaluate the exact forward model on retained candidates for a downstream task that explicitly accounts for the discarded region.

\subsection{Gaussian localization and feature selection}
\label{sec:gaussian-adaptive}

The compact Legendre rate does not transfer to the unfiltered Hermite
kernel. We therefore analyze a filtered Nevai extension for Gaussian
features, without changing the Christoffel ratio or claiming a rate for its
Gaussian version.

Assume that the independent standard Gaussian coordinates $(\xi_j)_{j\ge1}$ generate the prior $\sigma$-field modulo $\mu_0$. Let $S$ be a selected set of $k$ such coordinates and identify
its coordinate space with $\mathbb R^S\cong\mathbb R^k$, equipped with the standard Gaussian product measure $\gamma_{\mathrm G,S}$. Let
$\alpha=(\alpha_j)_{j\in S}\in\mathbb N_0^S$, and let $h_\alpha$ be the
orthonormal probabilists' Hermite products in $L^2(\gamma_{\mathrm G,S})$.  The vector
$\mathbf{e}_j$ below denotes the unit vector associated with coordinate $j\in S$.
Fix a finite set
$\Lambda\subset\mathbb N_0^S$ containing $0$, put
$D=\max_{\alpha\in\Lambda}|\alpha|$, and let
$\boldsymbol\rho=(\rho_j)_{j\in S}\in(0,1)^S$. Define
\begin{equation*}
K_{\boldsymbol{\rho},\Lambda}(\bxi,\bzeta)
 =\sum_{\alpha\in\Lambda}\boldsymbol{\rho}^\alpha h_\alpha(\bxi)h_\alpha(\bzeta),
 \qquad
 \boldsymbol{\psi}_{\boldsymbol{\rho},\Lambda}(\bxi)
 =(\boldsymbol{\rho}^\alpha h_\alpha(\bxi))_{\alpha\in\Lambda},
\end{equation*}
where $\boldsymbol{\rho}^\alpha=\prod_{j\in S}\rho_j^{\alpha_j}$.
Define $\Xi_S(q)=(\xi_j(q))_{j\in S}$, $\mathcal F_S=\sigma(\Xi_S)$, and $L_S=\E_{\mu_0}[L\mid\mathcal F_S]$. Let $f_S:\mathbb R^S\to[0,1]$ be a measurable version satisfying $L_S=f_S\circ\Xi_S$. Define the population score
\begin{equation}
 T_{\boldsymbol{\rho},\Lambda,S}(\bxi)
 =\frac{\boldsymbol{\psi}_{\boldsymbol{\rho},\Lambda}(\bxi)^\top H_{\Lambda,S}
                \boldsymbol{\psi}_{\boldsymbol{\rho},\Lambda}(\bxi)}
              {\|\boldsymbol{\psi}_{\boldsymbol{\rho},\Lambda}(\bxi)\|_2^2},\quad
 H_{\Lambda,S}
 =\E_{\mu_0}[L\,\boldsymbol{\phi}_{\Lambda,S}\boldsymbol{\phi}_{\Lambda,S}^\top],
 \label{eq:filtered-nevai-score}
\end{equation}
where $\boldsymbol{\phi}_{\Lambda,S}(q)=(h_\alpha(\Xi_S(q)))_{\alpha\in\Lambda}$.
The matrix is the same likelihood-weighted moment construction as before,
restricted to the chosen polynomial space.  Only the evaluation vector is
filtered.  In particular, the score has a normalized squared-kernel
representation and lies in $[0,1]$.

For later use, put
\begin{equation*}
b_j=\frac{2\rho_j}{1+\rho_j^2},\quad
 v_j=\frac{1-\rho_j^2}{1+\rho_j^2},\quad
 \delta(\rho_j)=(1-b_j)^2+v_j,\quad
 \mathcal A_\Lambda(\boldsymbol{\rho})
 =\sum_{\alpha\notin\Lambda}\boldsymbol{\rho}^{2\alpha}.
\end{equation*}
The last series is finite for each fixed $S$ and $\boldsymbol{\rho}\in(0,1)^k$.
The infinite filtered Hermite kernel is the classical Mehler kernel
\cite{Janson1997}; its normalized square is the Gaussian law
$\mathcal N((b_j\xi_j)_{j\in S},\operatorname{diag}(v_j))$.
We use this identity only to establish the following finite-polynomial
approximation estimate.

\begin{theorem}[Localized finite Hermite approximation]
\label{thm:gaussian-localized}
Suppose $0\le f_S\le1$ and, for some bounded linear map
$B_S:\mathbb R^k\to\mathcal Y$ into a Hilbert space,
\begin{equation}
 |f_S(\bxi)-f_S(\bzeta)|\le\|B_S(\bxi-\bzeta)\|_{\mathcal Y}, 
 \qquad \bxi,\bzeta\in\mathbb R^k.
 \label{eq:weighted-feature-lipschitz}
\end{equation}
Then the finite filtered score satisfies
\begin{equation}
 \|T_{\boldsymbol{\rho},\Lambda,S}-f_S\|_{L^2(\gamma_{\mathrm G,S})}
 \le
 \underbrace{\left[\sum_{j\in S}\|B_S\mathbf{e}_j\|_{\mathcal Y}^2
                     \delta(\rho_j)\right]^{1/2}}_{\mathfrak B(S,\boldsymbol{\rho})}
 +\sqrt{\mathcal A_\Lambda(\boldsymbol{\rho})}.
 \label{eq:gaussian-localization-bound}
\end{equation}
The first term measures Gaussian localization along the selected directions;
the second is the finite-polynomial truncation error.
\end{theorem}
\noindent\emph{The proof is given in Appendix~\ref{app:gaussian-localized}.}

The filtered construction has a different interpretation for the two scores:
if a diagonal filter is invertible on a \emph{fixed} polynomial space,
the genuine Christoffel ratio computed with the transformed prior and
likelihood Gram matrices is unchanged by this basis transformation.
Thus \eqref{eq:gaussian-localization-bound} is a Nevai result, not an
unproved Gaussian Christoffel convergence assertion.

\paragraph{Independent-pilot feature selection.}
\label{sec:gaussian-feature-selection}

Feature selection is formulated within the same prior-to-posterior geometry.
The need to control effective dimension is visible even with an exact
moment matrix: for $L(\bxi)=\tfrac12+\tfrac14\tanh(\xi_1)$ and the
first-degree Hermite basis on the first $m$ coordinates, direct calculation shows that 
\[
T_{1,m}(\bxi)=\frac12+\frac{2\mathfrak b\xi_1}{1+\sum_{j=1}^m\xi_j^2},
\]
where $\mathfrak b=\tfrac14\E[\xi_1\tanh(\xi_1)]>0$.
Thus $T_{1,m}\circ\Xi_m\to\tfrac12$ in $L^2(\mu_0)$ as $m\to\infty$, despite
$L$ depending only on the first coordinate.  This illustrates an
irrelevant-feature dilution effect at fixed polynomial degree, not a
failure of the function-space posterior itself.

For the Gaussian coordinate sequence $(\xi_j)_{j\ge1}$, let $\alpha$ range over finitely supported nonnegative integer sequences. Here $h_\alpha(q)$ abbreviates $\prod_j h_{\alpha_j}(\xi_j(q))$, where $h_n$ is the univariate normalized probabilists' Hermite polynomial. Write the orthogonal
Wiener--Hermite expansion $L=\sum_\alpha c_\alpha h_\alpha$ in $L^2(\mu_0)$,  $c_\alpha=\E_{\mu_0}[Lh_\alpha]$.  Whenever $\alpha$ belongs to the current polynomial index set,
$c_\alpha$ is the entry in column $0$ and the row corresponding to $h_\alpha$ in the likelihood-weighted Hankel matrix. For a finite coordinate set $S$, put $\eta_S:=\|L-L_S\|_{L^2(\mu_0)}$ using the conditional likelihood defined above. By orthogonality,
\begin{equation}
 \eta_S^2
 =\sum_{\operatorname{supp}\alpha\not\subseteq S}c_\alpha^2.
 \label{eq:feature-selection-energy}
\end{equation}
Hence this criterion measures lost \emph{likelihood information}, rather
than prior coordinate variance.

Fix a finite candidate Hermite index set $\Gamma$ containing $0$,
and restrict admissible coordinate sets $S$ to subsets of the finite pool
$J_\Gamma=\bigcup_{\alpha\in\Gamma}\operatorname{supp}\alpha$.
From $M_{\rm sel}$ iid prior draws construct
$\widehat c_\alpha=M_{\rm sel}^{-1}\sum_iL(Q_i)h_\alpha(Q_i)$.
Given a coordinate budget $k$, choose $\widehat S$ minimizing
\begin{equation}
 \widehat r_\Gamma(S)
 =\left(\sum_{\substack{\alpha\in\Gamma\\
                    \operatorname{supp}\alpha\not\subseteq S}}
                      \widehat c_\alpha^2\right)^{1/2},
 \quad S\subseteq J_\Gamma,\quad |S|\le k.
 \label{eq:adaptive-selector}
\end{equation}
An exact minimizer is used for the mathematical statement; the cost of
solving this combinatorial selection problem is not asserted to be small.

\begin{theorem}[Feature-selection and conditional Nevai error]
\label{thm:adaptive-hankel-nevai}
Let $\varepsilon_\Gamma=\|L-P_\Gamma L\|_2$, where $P_\Gamma$
denotes Hermite projection onto $\operatorname{span}\{h_\alpha:\alpha\in\Gamma\}$.
Then the selector \eqref{eq:adaptive-selector} satisfies
\begin{equation}
 \E_{\rm sel}\eta_{\widehat S}
 \le \inf_{\substack{S\subseteq J_\Gamma\\|S|\le k}}\eta_S
       +\varepsilon_\Gamma
       +2\sqrt{\frac{|\Gamma|}{M_{\rm sel}}}.
 \label{eq:feature-selection-bound}
\end{equation}
Conditionally on the first pilot, let $\boldsymbol{\rho},\Lambda$ be fixed by a
measurable selection rule on $\widehat S$, with
$D=\max_{\alpha\in\Lambda}|\alpha|$ bounded by a predetermined
integer $D_{\max}$. Assume the conditional likelihood $f_{\widehat S}$
satisfies \eqref{eq:weighted-feature-lipschitz} with $B_{\widehat S}$.
Use an independent second pilot of size $M_{\rm est}$ to construct the
empirical version $\widehat T_{\boldsymbol{\rho},\Lambda,\widehat S}$ of
\eqref{eq:filtered-nevai-score}. Then
\begin{equation}
\begin{aligned}
 \E\|\widehat T_{\boldsymbol{\rho},\Lambda,\widehat S}\circ\Xi_{\widehat S}-L\|_{L^2(\mu_0)}
 &\le\inf_{\substack{S\subseteq J_\Gamma\\|S|\le k}}\eta_S
       +\varepsilon_\Gamma+2\sqrt{\frac{|\Gamma|}{M_{\rm sel}}}\\
 &\quad+\E_{\rm sel}\!\left[\mathfrak B(\widehat S,\boldsymbol{\rho})
              +\sqrt{\mathcal A_\Lambda(\boldsymbol{\rho})}\right]
       +\frac{3^{D_{\max}}}{\sqrt{M_{\rm est}}}.
\end{aligned}
 \label{eq:adaptive-master-bound}
\end{equation}
In the last estimate, the empirical score is composed with $\Xi_{\widehat S}$, so both terms are functions on the original prior space.
\end{theorem}
\noindent\emph{The proof is given in Appendix~\ref{app:adaptive-hankel-nevai}.}

By Theorem~\ref{thm:fixed-budget-regret} and
$\|L-\widehat T\|_1\le\|L-\widehat T\|_2$, the right-hand side of
\eqref{eq:adaptive-master-bound}, divided by $Z$, also bounds the expected
posterior-mass regret of the adaptive score at any fixed retained prior
fraction $0<\alpha<1$ (with ties resolved measurably).

\paragraph{Scope and computational cost.}
For $s=|\Lambda|$, assembling a dense empirical moment matrix by outer
products costs $O(M_{\rm est}s^2)$ arithmetic operations after the basis
has been evaluated, storing it costs $O(s^2)$, and evaluating one Nevai
quadratic form costs $O(s^2)$ (or less if structure is exploited).
The selector \eqref{eq:adaptive-selector} is combinatorial; no claim of
optimal computational complexity is made. The principal inverse-problem experiments compare the original scores and, where stated, fixed-filter variants. Section~\ref{sec:num-filtered} separately tests filtered localization and the independent-pilot selector on controlled Gaussian examples.

\section{Numerical experiments}
\label{sec:numerics}

The main experiments comprise a multimodal geometry check, a nonlinear PDE comparison against regression surrogates, and controlled Gaussian tests of the filtered extension. Additional inverse-source and diffusion diagnostics are provided below. Every pilot uses deterministic likelihood values. Reference likelihoods are used only for retrospective validation and are never supplied to fitting, hyperparameter selection, or candidate ranking by a learned method. Posterior-mass capture is a ratio of weighted sums, not an exact posterior probability. We distinguish variability over independent pilots from uncertainty due to the finite reference cloud.
For evaluation candidates $Q_1^{\rm ref}, \ldots, Q_N^{\rm ref}$ and a retained index set $I_k(s)$ containing the $k$ largest scores, we report
\begin{equation}
 \widehat C_s(k/N)
 =\frac{\sum_{i\in I_k(s)}L(Q_i^{\rm ref})}
        {\sum_{i=1}^N L(Q_i^{\rm ref})}.
 \label{eq:empirical-capture}
\end{equation}
Ties are resolved using candidate indices without consulting their likelihoods. Exact-likelihood ranking takes $s_i=L(Q_i^{\rm ref})$ and maximizes this empirical capture at each $k$; it is a validation benchmark requiring all candidate forward solves. Random selection of $k$ indices has expected capture $k/N$. Unless otherwise stated, Nevai and Christoffel values are clipped to $[0,1]$ before ranking; clipping can create ties, which are handled by the same rule.

\subsection{A nonlinear multimodal inverse problem}

Let \(q=(q_1,q_2)\in[-1,1]^2\) have the uniform prior and consider the nonlinear forward map
\begin{equation}
  \mathcal G(q)=(q_1^2,q_2^2).
  \label{eq:generic-forward}
\end{equation}
We use data $y=(0.498,0.244)$, which are a small perturbation of \(\mathcal G(0.7,-0.5)=(0.49,0.25)\).  Because \eqref{eq:generic-forward} loses both signs, the posterior has four separated modes.  This makes the example deliberately unfavorable to a single local Gaussian approximation and tests whether a global polynomial moment score can recover disconnected posterior-relevant regions.

For the sharp screening experiment we take independent Gaussian noise with \(\sigma=0.07\), so
$L(q; y)=\exp\!\left[-\frac{\|\mathcal G(q)-y\|_2^2}{2\sigma^2}\right]$.
A degree-four Legendre basis has \(s_d=15\). All Christoffel computations in this example use the same diagonal ridge $10^{-9}$, and the independent evaluation cloud contains $N=120000$ candidates.  The empirical Hankel matrix is built from \(M=20000\) iid prior candidates using their deterministic likelihood values as weights.

Figure~\ref{fig:general-fields} compares the reference likelihood with the empirical Nevai and Christoffel score fields.  For visualization, each field is min--max rescaled to $[0,1]$ and overlaid with contours; the likelihood omits the data-only Gaussian normalization constant. Both scores recover all four modes.  Table~\ref{tab:general-screening} reports quantitative ranking results: $c_{10}$ and $c_{20}$ denote capture at $10\%$ and $20\%$ retention, while $\mathrm{Corr}_{\rm S}^{\rm N}$ and $\mathrm{Corr}_{\rm S}^{\rm C}$ denote Spearman correlations with the reference likelihood. The comparison uses either the $M=20000$ pilot matrix or a $120\times120$ Gauss--Legendre population-matrix approximation.  The top \(20\%\) of prior candidates selected by the degree-four Nevai score contains about \(99.5\%\) of the reference posterior mass.

\begin{figure}[t]
\centering
\includegraphics[width=0.98\textwidth]{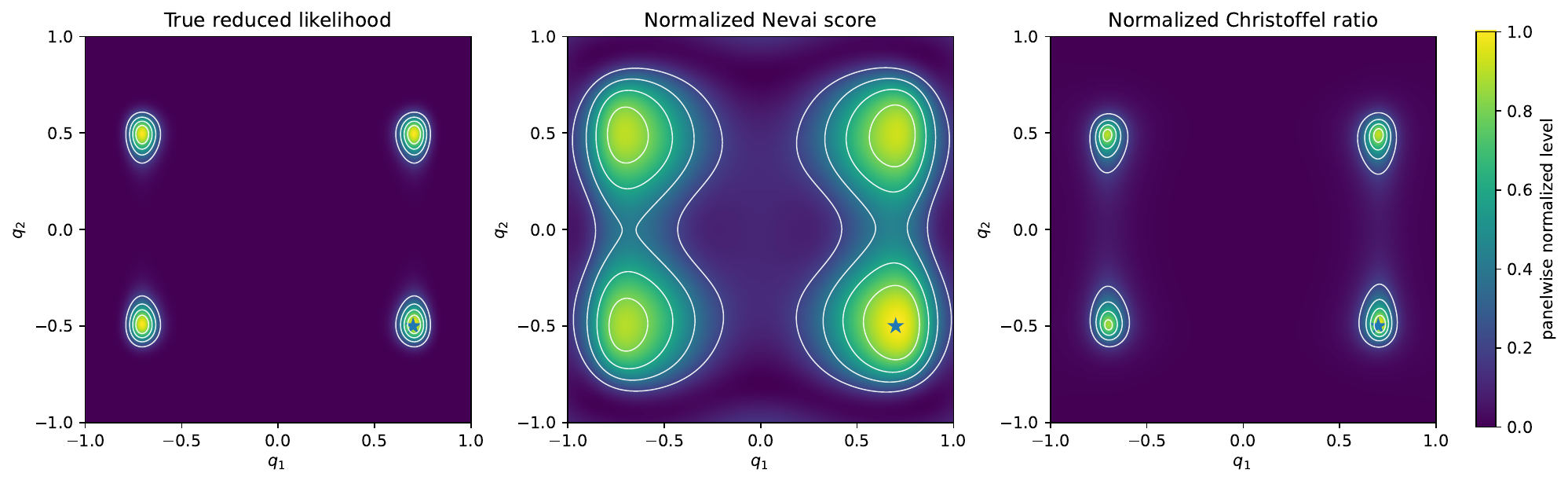}
\caption{Multimodal example: reference likelihood, Nevai score, and Christoffel ratio. Panels are independently rescaled to $[0,1]$; the star marks the generating parameter.}
\label{fig:general-fields}
\end{figure}

\begin{table}[t]
\centering
\footnotesize
\caption{Multimodal screening: Spearman correlations and captured reference mass at $d=4$.}
\label{tab:general-screening}
\begin{tabular}{lcccccc}
\toprule
pilot & $\mathrm{Corr}_{\rm S}^{\rm N}$ & $\mathrm{Corr}_{\rm S}^{\rm C}$ & \(c_{10}^{\rm N}\) & \(c_{20}^{\rm N}\) & \(c_{10}^{\rm C}\) & \(c_{20}^{\rm C}\)\\
\midrule
likelihood-weighted pilot & 0.961 & 0.930 & 0.965 & 0.995 & 0.974 & 0.999\\
quadrature reference & 0.963 & 0.928 & 0.969 & 0.995 & 0.974 & 0.999\\
\bottomrule
\end{tabular}
\end{table}

Figure~\ref{fig:general-curves} extends this comparison across retained fractions. The exact-likelihood ranking is the upper benchmark for this reference cloud; both learned rankings approach it as the retained fraction increases.

Figure~\ref{fig:general-convergence} examines degree and pilot resolution separately. For the degree sweep we use a smoother likelihood, $\sigma=0.20$, and a fixed pilot of $M=50000$. The plotted total-variation distance compares the normalized surrogate measure $\widetilde T_{d, M}\mu_0/\int\widetilde T_{d, M}\,\dd\mu_0$ with $\mu^y$, using $120\times120$ quadrature. It measures a reweighting approximation, not the sampling law of retained candidates. For the pilot sweep, $\sigma=0.07$ and $d=4$ are fixed; operator errors relative to the quadrature matrix are averaged over 20 pilots per size. These finite-range diagnostics illustrate improvement without identifying an asymptotic rate.
\begin{figure}[!htbp]
\centering
\begin{subfigure}{\textwidth}
\centering
\includegraphics[width=0.4\textwidth]{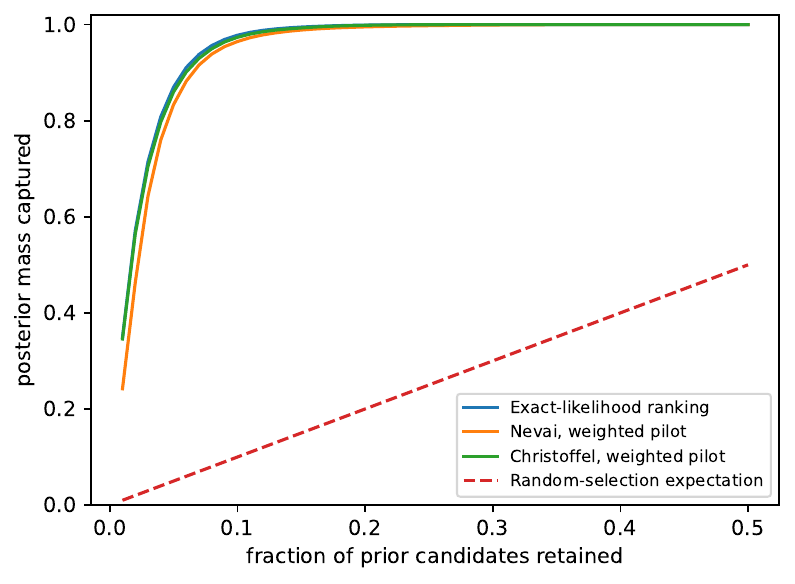}
\caption{Captured mass versus retained fraction ($\sigma=0.07$, $d=4$, $M=20000$). The dashed line denotes random-selection expectation.}
\label{fig:general-curves}
\end{subfigure}
\smallskip
\begin{subfigure}{\textwidth}
\includegraphics[width=0.8\textwidth]{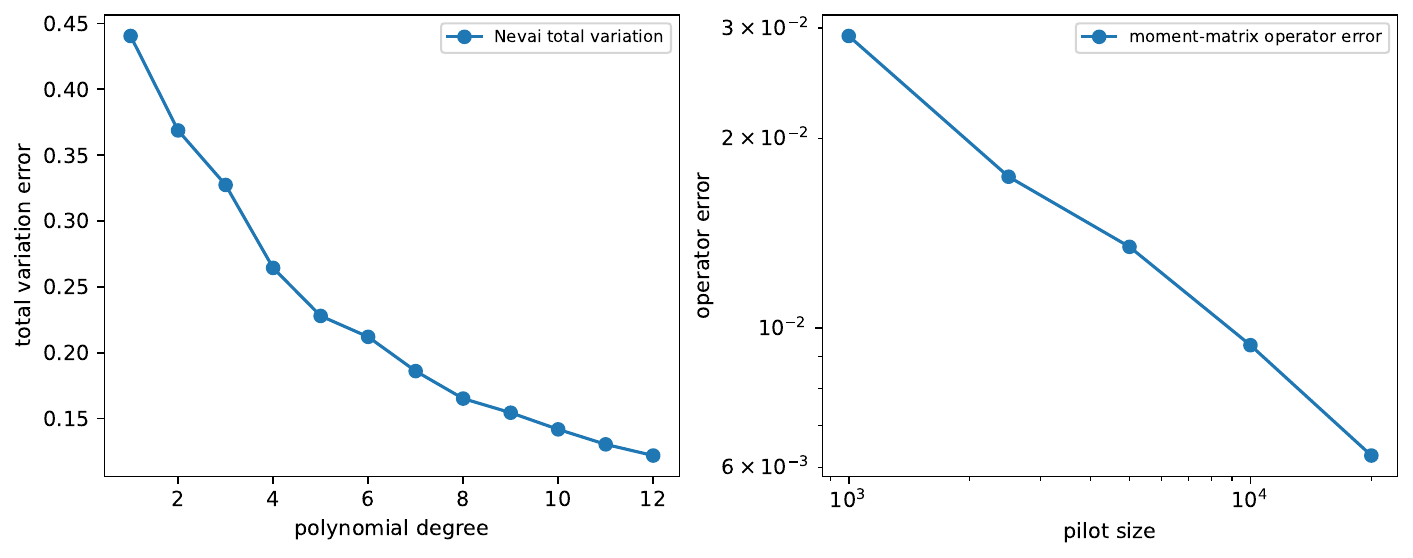}
\caption{Resolution diagnostics: total-variation error versus degree (left) and mean matrix error versus pilot size (right).}
\label{fig:general-convergence}
\end{subfigure}
\smallskip
\caption{Multimodal screening and resolution diagnostics.}
\end{figure}

\FloatBarrier
\subsection{A nonlinear PDE with a function-valued reaction coefficient}
\label{sec:num-nonlinear-function}
To test a function-valued unknown beyond an analytically solvable forward map, consider
\begin{equation}
 -u''(x)+6\exp(q(x))u(x)^3=30\sin(\pi x),\qquad
 u(0)=u(1)=0,\qquad x\in(0,1),
 \label{eq:nonlinear-function-pde}
\end{equation}
with five observations $u(s_j)$ at $s=(0.18,0.32,0.50,0.68,0.82)$ and independent Gaussian noise of standard deviation $0.045$. The numerical prior is a 24-mode Gaussian Fourier field,
\begin{equation}
 q(x)=\sum_{k=1}^{24}0.90\,k^{-0.7}\xi_k\sqrt2\sin(k\pi x),
 \qquad \xi_k\stackrel{\rm iid}{\sim}\gamma_{\mathrm G}.
 \label{eq:nonlinear-function-prior}
\end{equation}
The truth uses the first nine coordinates $(1,-0.85,0.70,-0.60,0.50,-0.42,0.34,-0.28,0.20)$ and zeros thereafter. Fixed observation perturbations are $(0.008,-0.009,0.004,-0.006,0.007)$. We solve the semilinear boundary value problem by centered differences on 64 interior nodes and damped Newton iteration to a maximum discrete residual below $10^{-6}$. This is a finite-resolution nonlinear PDE benchmark; the reported numerical scores do not claim exact continuum inference.

\paragraph{Common budgets and independent validation.}
We use $N=120000$ independent reference prior draws and five independent pilots, with nested sizes $M\in\{1000,4000,8000\}$ within each replicate. The likelihood standard deviations are $\sigma=0.045$ and $0.09$, with the same fixed data and forward outputs; this varies posterior concentration, not the data realization. All learned scores use the first $m=8$ coordinates. The total-degree-two Hermite space has dimension 45. The Christoffel ridge is $10^{-10}$ on the unnormalized moment matrix. Filtered Nevai scores use preassigned $\rho=0.6$ or $0.85$ in every coordinate. These filters are not selected using reference results.

A 64-to-128-node check on 100 independent draws changes any sensor output by at most $6.5\times10^{-4}$, compared with the smaller noise standard deviation $0.045$. It does not establish continuum or parameter-truncation convergence. The reference importance effective sample sizes are 1491 and 8720 for the two noise levels, respectively. All comparisons use the same reference cloud and matched pilot draws.

\paragraph{Regression competitors and information access.}
For both total-degree-two polynomial ridge regression and a rank-128 Nystr\"om approximation to the Gaussian radial-basis kernel \cite{WilliamsSeeger2001}, we fit three targets: $L$, $\log L$, and the vector of five forward outputs. Log-likelihoods are computed stably from the residual, including where exponentiation underflows. Forward surrogates use the same solves but more than their scalar likelihood summaries. All regressors use an unpenalized intercept. On a fixed 80/20 split within each pilot, the sum-loss ridge parameter is chosen from $\{10^{-6},10^{-3},1,100\}$; kernel bandwidth parameters in $\exp(-\vartheta\|\bxi-\bzeta\|_2^2)$ are chosen from $\{0.025,0.1,0.4\}$. Validation minimizes squared error for the fitted target, with no reference likelihoods. Nystr\"om landmarks are sampled from the training split during tuning and from the full pilot for refitting. Each chosen model is refitted using all $M$ pilot draws. These are specified, reproducible baselines, not an exhaustive comparison with surrogate modelling.

Likelihood regressions are clipped to $[0,1]$. Log-likelihood predictions are ranked directly. Forward surrogates are ranked by their Gaussian log-likelihood with the stated noise level. Thus likelihood underflow does not artificially tie their candidate rankings.

\begin{table}[htbp]
\centering\small
\caption{Nonlinear PDE: captured mass at $\sigma=0.045$, $M=4000$, $m=8$; mean $\pm$ sample standard deviation over five pilots.}
\label{tab:nonlinear-function}
\resizebox{\textwidth}{!}{\begin{tabular}{llcc}
\toprule
Score & Pilot information & Top $5\%$ & Top $20\%$\\
\midrule
% BEGIN_JUQ_comparison
Christoffel & Scalar likelihood & $0.92177\pm 0.00710$ & $0.99992\pm 0.00003$ \\
Nevai & Scalar likelihood & $0.42165\pm 0.01612$ & $0.82461\pm 0.03466$ \\
Clipped projection & Scalar likelihood & $0.23592\pm 0.05173$ & $0.80509\pm 0.03054$ \\
Filtered Nevai, $\rho=0.6$ & Scalar likelihood & $0.42216\pm 0.01995$ & $0.83533\pm 0.03528$ \\
Filtered Nevai, $\rho=0.85$ & Scalar likelihood & $0.42633\pm 0.01566$ & $0.83056\pm 0.03451$ \\
Polynomial fit to $L$ & Scalar likelihood & $0.23214\pm 0.05353$ & $0.80047\pm 0.03045$ \\
Polynomial fit to $\log L$ & Scalar likelihood & $0.07760\pm 0.01565$ & $0.91839\pm 0.02957$ \\
Kernel fit to $L$ & Scalar likelihood & $0.38265\pm 0.06362$ & $0.78091\pm 0.05635$ \\
Kernel fit to $\log L$ & Scalar likelihood & $0.28672\pm 0.05299$ & $0.94697\pm 0.01357$ \\
Polynomial forward surrogate & Forward outputs & $0.95935\pm 0.00085$ & $0.99998\pm 1.66\!\times\!10^{-6}$ \\
Kernel forward surrogate & Forward outputs & $0.97097\pm 0.00186$ & $0.99987\pm 0.00024$ \\
% END_JUQ_comparison
\bottomrule
\end{tabular}}
\end{table}

\paragraph{Screening performance and uncertainty.}
Table~\ref{tab:nonlinear-function} and Figure~\ref{fig:nonlinear-function} show that Christoffel captures $0.922$ of the reference mass at $5\%$ retention for $\sigma=0.045$, compared with $0.422$ for Nevai. The tested direct likelihood regressions are less effective at this budget. However, polynomial and kernel forward surrogates capture $0.959$ and $0.971$. At $\sigma=0.09$, the corresponding Christoffel and forward-surrogate values are $0.626$, $0.682$, and $0.689$. Thus the comparison does not establish a general advantage over forward surrogates. Fitting a scalar likelihood, fitting its logarithm, and fitting the forward response are materially different approximation tasks even under the same solve budget.

For reference uncertainty, we use 200 paired Poisson bootstrap resamples of reference candidates, keeping the fitted pilots fixed and recomputing ranks, integer budgets, and self-normalized mass estimates. The percentile interval for the mean over the five fixed Christoffel fits at $5\%$ retention is $[0.9145,0.9279]$, while that for the kernel forward surrogate is $[0.9667,0.9749]$. The paired difference has interval $[0.0440,0.0540]$. These intervals assess reference Monte Carlo variability conditionally on the fitted scores; they are not joint confidence intervals over both pilot and reference sampling. The table reports pilot variability separately; values rounded close to one do not imply exact recovery of posterior mass. The shaded bands and error bars in the comparison figures likewise represent pilot variability, not reference uncertainty. Neither uncertainty measure includes PDE discretization, finite-mode approximation, or a change of observed data.

Figure~\ref{fig:juq-pilot} varies pilot size without changing the reference cloud. Fixed-degree filtering does not remove the main Nevai limitation here: at $M=4000$, the two filters capture $0.422$ and $0.426$ at $5\%$, close to the unfiltered value. A Gaussian localization theorem concerns a coordinated choice of filter and polynomial resolution; it does not guarantee that filtering a degree-two score improves its ranking.

\begin{figure}[!htbp]
\centering
\begin{subfigure}{\textwidth}
\centering
\includegraphics[width=0.88\textwidth]{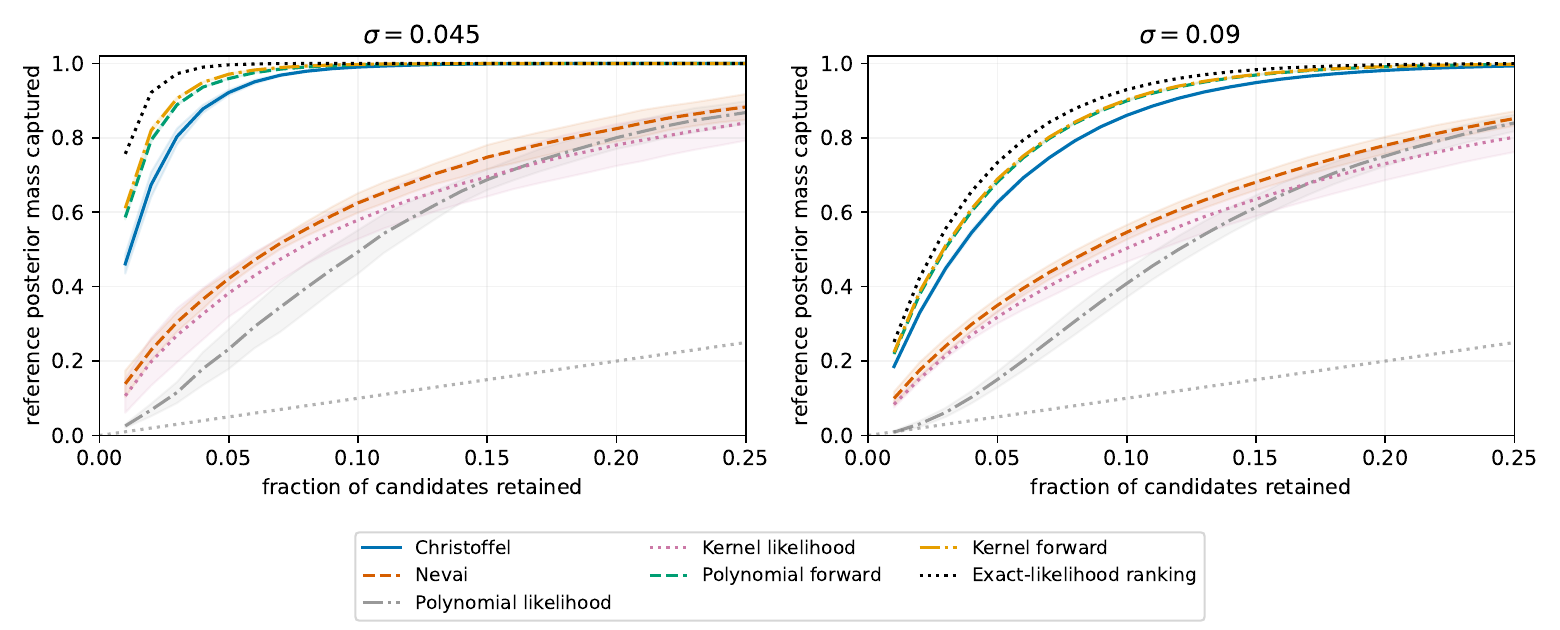}
\caption{Fixed pilot budget $M=4000$: mean capture and one-standard-deviation bands over five pilots.}
\label{fig:nonlinear-function}
\end{subfigure}
\smallskip
\begin{subfigure}{\textwidth}
\includegraphics[width=0.88\textwidth]{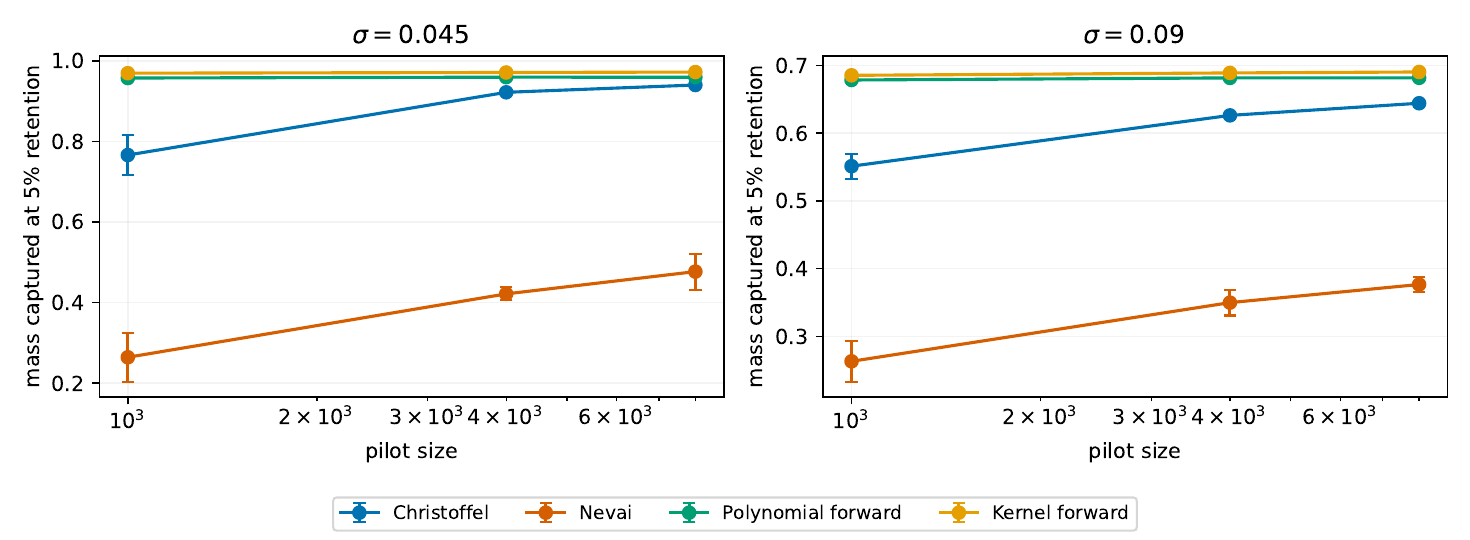}
\caption{Varying pilot budget: capture at $5\%$ retention; error bars show one pilot standard deviation.}
\label{fig:juq-pilot}
\end{subfigure}
\smallskip
\caption{Nonlinear PDE screening at fixed and varying pilot budgets.}
\end{figure}

\paragraph{Forward-solve accounting and measured costs.}
Screening $N$ new candidates and then evaluating the $k$ retained candidates requires $M+k$ forward solves, in addition to fitting and scoring. A simple cost comparison is
\begin{equation}
 (M+k)c_{\rm f}+t_{\rm fit}+t_{\rm score}+t_{\rm sort}<Nc_{\rm f},
 \label{eq:screening-cost}
\end{equation}
where $c_{\rm f}$ is an average solve cost. This is an accounting condition, not an accuracy guarantee; expensive feature extraction must also be charged. In Table~\ref{tab:juq-cost}, scoring includes feature evaluation for every method. For moment scores, the shared Hermite construction is charged separately to each standalone method. Fitting includes internal tuning for regressors. Timings use one BLAS thread and the first pilot; solve times are medians of three runs. Total time includes fitting, scoring, sorting, pilot solves, and retained-candidate solves. The break-even solve cost follows from \eqref{eq:screening-cost}. Reference solves used solely to evaluate the experiment are excluded from deployment costs.

\begin{table}[htbp]
\centering\small
\caption{Nonlinear PDE costs for $M=4000$, $N=120000$, $k=6000$. Times are in seconds; break-even solve costs are in microseconds.}
\label{tab:juq-cost}
\begin{tabular}{lrrrr}
\toprule Method & Fit & Score & Total & Break-even\\\midrule
% BEGIN_JUQ_cost
Christoffel & 0.003 & 0.144 & 0.465 & 1.44 \\
Nevai & 0.003 & 0.104 & 0.434 & 1.07 \\
Polynomial likelihood & 0.008 & 0.047 & 0.390 & 0.59 \\
Kernel likelihood & 0.066 & 0.153 & 0.552 & 2.08 \\
Polynomial forward & 0.008 & 0.057 & 0.381 & 0.69 \\
Kernel forward & 0.069 & 0.172 & 0.559 & 2.29 \\
% END_JUQ_cost
\bottomrule
\end{tabular}
\end{table}

On this small batched solver, evaluating all reference candidates takes about $2.60$ seconds. The assembled Christoffel screening-and-retained-evaluation cost is $0.465$ seconds and uses 10000 instead of 120000 solves. The polynomial forward surrogate is faster here ($0.381$ seconds) and captures more mass. These measurements support a conditional saving in forward evaluations, not superiority in total runtime or exact posterior sampling. They exclude neither pilot construction nor feature evaluation, but remain hardware- and implementation-dependent. A retained subset can omit posterior mass even when its evaluation is inexpensive.

\subsection{Controlled filtered Gaussian localization and feature selection}
\label{sec:num-filtered}
The following tests isolate the filtered construction from PDE and pilot effects. First, take the prior $\gamma_{\mathrm G}$ and the likelihood $L(\xi)=\exp[-(\xi-1)^2/(2\tau^2)]$, $\tau=0.5$. For $\Lambda=\{0,\ldots,d\}$, the infinite filtered score has the explicit form
\[
 T_{\rho,\infty}(\xi)=\frac{\tau}{\sqrt{\tau^2+v}}
 \exp\!\left[-\frac{(b\xi-1)^2}{2(\tau^2+v)}\right].
\]
Here $b$ and $v$ are the one-coordinate specializations of $b_j$ and $v_j$ above. We compute population moment matrices and $L^2(\gamma_{\mathrm G})$ errors by 240-node Gauss--Hermite quadrature, checking against 160 nodes. The maximum change in the reported likelihood-approximation errors is $1.1\times10^{-7}$. For $d\in\{2,4,8,16,32,64\}$ and $\rho\in\{0.5,0.8,0.95\}$, the finite-to-infinite errors satisfy the numerical checks against $\rho^{d+1}/\sqrt{1-\rho^2}$, and the total errors satisfy Theorem~\ref{thm:gaussian-localized} with Lipschitz constant $e^{-1/2}/\tau$. These checks illustrate, rather than prove, the bounds. At fixed $\rho<1$, increasing $d$ approaches a nonzero localization-error floor, as shown in Figure~\ref{fig:filtered-validation}; increasing degree alone does not imply consistency.

Second, consider eight independent Gaussian coordinates and the bounded likelihood $L(\bxi)=\tfrac12+0.15\sin\xi_1+0.15\sin\xi_5$. With $\Gamma=\{0,\mathbf e_1,\ldots,\mathbf e_8\}$ and a two-coordinate budget, minimizing \eqref{eq:adaptive-selector} amounts to retaining the two largest squared empirical linear coefficients. We use independent selection and estimation pilots of 2000 draws each, 20 independent pilot pairs, and 40000 independent reference candidates. The selection pilot identifies $\{1,5\}$ in all 20 runs; this is an observation on this test, not a universal selection guarantee. Degree-two scores on the selected coordinates, the first two coordinates, and all eight coordinates use the same estimation pilot. The fixed-coordinate baselines therefore leave the 2000 selection evaluations unused; this test isolates the feature choice rather than optimizing each method's total sample allocation.

In Figure~\ref{fig:filtered-validation}, $\rho=1$ denotes the unfiltered comparator, and feature-selection results are means with one standard deviation over the 20 runs. At $\rho=0.6$, mean $L^2$ errors are $0.0582$ for selected coordinates, $0.1076$ for the first two, and $0.0889$ for all eight. For the unfiltered selected score the error is $0.0626$. Thus relevant-feature selection has a larger effect than the filter change in this example. The experiment tests the specified finite candidate selector and independent-pilot construction; it does not establish dimension-independent performance or validate the potentially loose $3^D/\sqrt{M}$ estimation bound sharply.

\begin{figure}[htbp]
\centering
\includegraphics[width=0.90\textwidth]{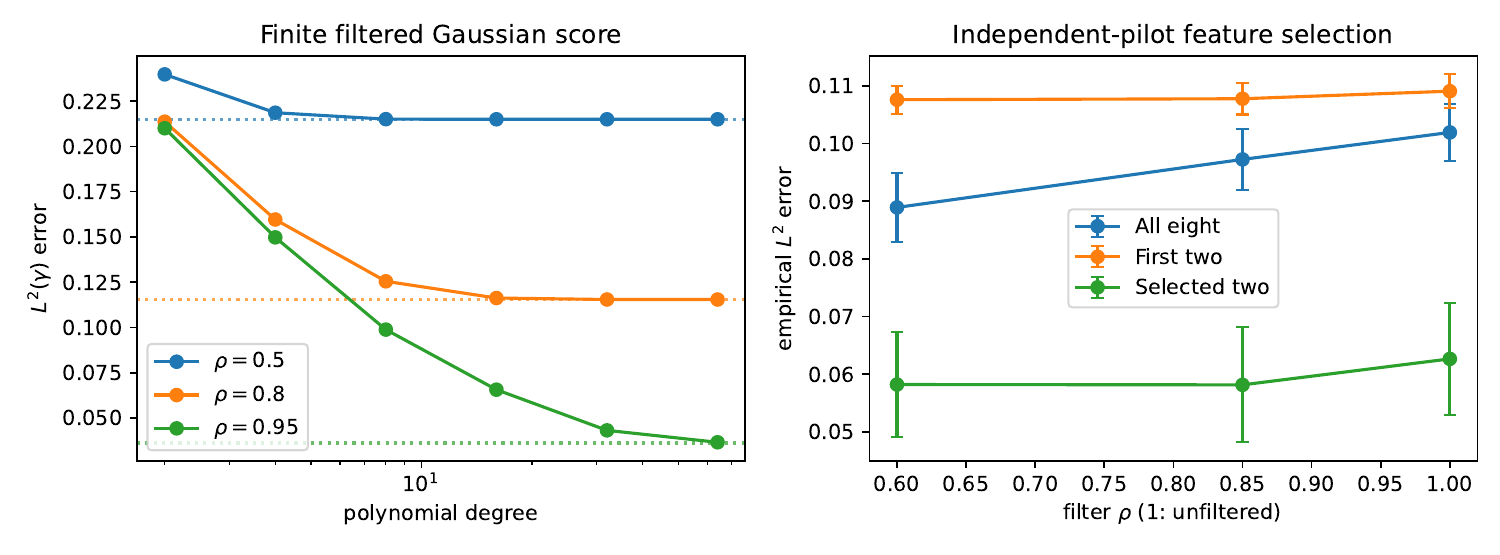}
\caption{Filtered Gaussian tests: population localization error (left) and independent-pilot feature selection (right). Dotted lines show error floors; error bars show one standard deviation.}
\label{fig:filtered-validation}
\end{figure}

\subsection{Function-space inverse source: feature resolution and screening}
\label{sec:num-function}

We next test the function-space theory. Consider
\[
-u''(x)=q(x),\qquad u(0)=u(1)=0,\qquad x\in(0,1).
\]
with Gaussian source prior
\begin{equation*}
q(x)=\sum_{k\ge1}\sqrt{\lambda_k}\,\xi_k\sqrt2\sin(k\pi x),
 \qquad
 \xi_k\stackrel{\rm iid}{\sim}\gamma_{\mathrm G},
 \qquad
 \lambda_k=0.7^2 k^{-1.3}.
\end{equation*}
We observe the flux $u'(s_j)$ at five interior sensors. Since
$u'(s_j)=\sum_{k\ge1} \frac{\sqrt{2\lambda_k}}{k\pi}\cos(k\pi s_j)\,\xi_k$,
the forward map is linear in the Gaussian coordinates. The numerical reference retains 96 modes; this is used only to evaluate the benchmark, not to define the screening theory. All reported finite-dimensional reference quantities, including the conditional likelihoods below, condition within this 96-mode model; they are not exact numerical evaluations of the infinite-series posterior.

The feature vector is $\bxi_{1:m}=\Xi_m(q)=(\xi_1(q),\ldots,\xi_m(q))^\top$. Because the omitted tail remains Gaussian and enters the observation linearly, the conditional likelihood is available analytically (with the tail truncated at mode 96 in the numerical benchmark). If $A_m$ is the observation matrix for the first $m$ coordinates and $\Sigma_{>m}=A_{>m}A_{>m}^\top$ is the omitted observation covariance, then
\begin{equation}
 \bar L_m(\bxi_{1:m})
 =\det\!\left(I+\frac{\Sigma_{>m}}{\sigma^2}\right)^{-1/2}
 \exp\!\left[
 -\frac{1}{2\sigma^2}\mathbf{r}_m^\top
 \left(I+\frac{\Sigma_{>m}}{\sigma^2}\right)^{-1}\mathbf{r}_m
 \right],
 \label{eq:conditional-linear-gaussian}
\end{equation}
where $\mathbf{r}_m=A_m\bxi_{1:m}-y$. This formula serves as a direct reference for Theorem~\ref{thm:feature-posterior}.

We use $\sigma=0.06$, an independent evaluation cloud of $N=60000$ draws, and nested prefixes of one $100000$-draw pilot for the pilot-size sweep. Sensor positions, the fixed data perturbation, and random seeds are specified in the reproduction scripts. The conditional likelihood is analytic, while its variance, relevance energies, and captured mass are estimated on the evaluation cloud. Table~\ref{tab:function-space} and the left panel of Figure~\ref{fig:function-space-screening} show the resolved fraction $\mathfrak R_m$ from \eqref{eq:resolved-fraction}. Eight coordinates resolve $92.3\%$ of the likelihood variance, twelve resolve $96.6\%$, and twenty-four resolve $99.3\%$. Screening is easier than global likelihood reconstruction: ranking candidates by the exact feature likelihood $\bar L_8$ and keeping only the top $20\%$ already captures $99.99\%$ of the 96-mode reference posterior mass in the evaluation cloud.

To test the empirical Hankel step, we use a total-degree-two Hermite cylinder basis at $m=8$. Its dimension is 45. With $M=50000$ weighted pilot candidates, the empirical Nevai score has Spearman correlation $0.763$ with the 96-mode reference likelihood and its top $20\%$ prior fraction captures $93.9\%$ of the 96-mode reference posterior mass. The center panel of Figure~\ref{fig:function-space-screening} shows that the mass capture improves as the pilot size grows. The difference from conditional-likelihood ranking may reflect both Nevai localization and finite-pilot errors in \eqref{eq:nevai-three-scale}; the experiment does not separately identify their contributions.

Finally, the right panel of Figure~\ref{fig:function-space-screening} plots the relevance energies $\mathcal E_m$, normalized by the sum of the first 24 estimates on the reference cloud. They decay rapidly but not in exact proportion to the prior eigenvalues. This is the intended interpretation of the spectrum: it measures how much a mode changes prior-to-posterior relevance, rather than how much prior variance the mode carries.

\begin{figure}[tbp]
\centering
\includegraphics[width=0.98\textwidth]{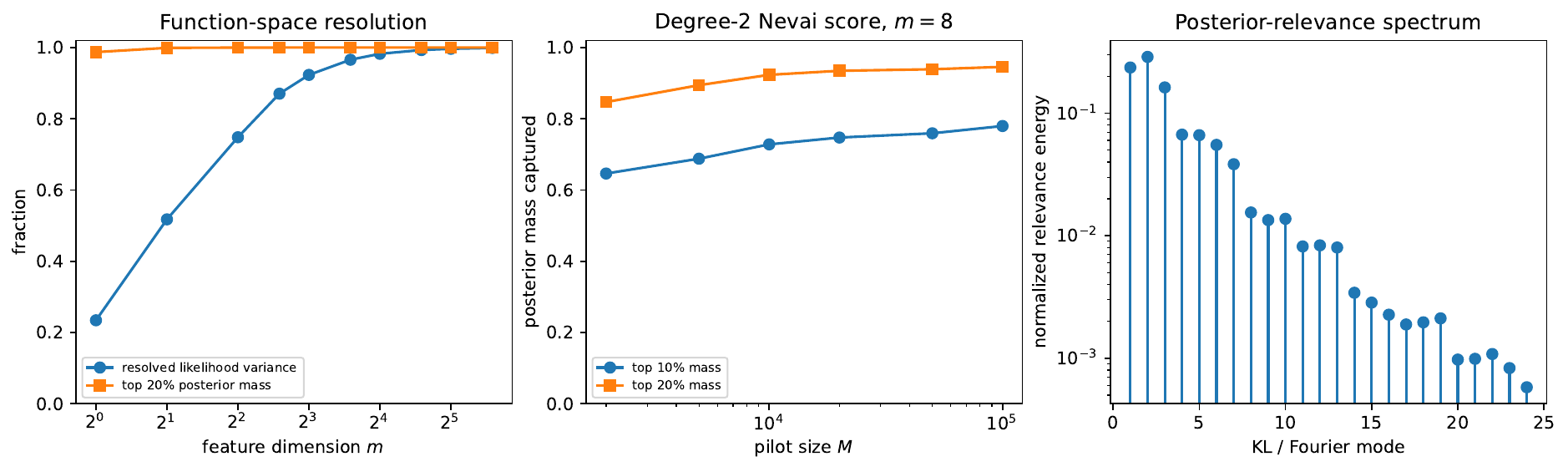}
\caption{Gaussian inverse-source model: feature resolution (left), pilot-size dependence (center), and normalized relevance energies (right).}
\label{fig:function-space-screening}
\end{figure}

\begin{table}[tbp]
\centering
\small
\caption{Gaussian inverse-source model: resolved likelihood variance and mass captured by exact feature-likelihood ranking.}
\label{tab:function-space}
\begin{tabular}{rccc}
\toprule
$m$ & resolved fraction $\mathfrak R_m$ & top 10\% mass & top 20\% mass\\
\midrule
2  & 0.52 & 0.95 & 0.999\\
4  & 0.75 & 0.98 & 0.9997\\
8  & 0.92 & 0.98 & 0.9999\\
12 & 0.97 & 0.99 & 0.9999\\
24 & 0.99 & 0.99 & 0.9999\\
\bottomrule
\end{tabular}
\end{table}

To expose the interaction among feature resolution, polynomial degree, and pilot size, Table~\ref{tab:function-space-grid} uses the first $M=10000$ draws of the same pilot and the same evaluation cloud, varying both $m$ and $d$.  The dependence on $d$ is not uniformly monotone in this finite-pilot experiment; at larger $m$, a larger polynomial space can increase the empirical moment-estimation burden.  The comparison is compatible with Theorem~\ref{thm:function-nevai}: adding features reduces conditional-likelihood error, but may increase both fixed-degree localization error and pilot error. This single-pilot grid does not separate those effects.

\begin{table}[tbp]
\centering
\small
\caption{Joint feature and polynomial resolution: Nevai mass capture at $20\%$ retention with $M=10000$.}
\label{tab:function-space-grid}
\begin{tabular}{rrrr}
\toprule
$m$ & degree $d$ & basis size $s_{d,m}$ & top $20\%$ mass\\
\midrule
4 & 1 & 5 & 0.896 \\
4 & 2 & 15 & 0.925 \\
4 & 3 & 35 & 0.943 \\
8 & 1 & 9 & 0.903 \\
8 & 2 & 45 & 0.923 \\
8 & 3 & 165 & 0.926 \\
12 & 1 & 13 & 0.878 \\
12 & 2 & 91 & 0.893 \\
12 & 3 & 455 & 0.876 \\
\bottomrule
\end{tabular}
\end{table}

The experiment therefore separates three ideas that are often mixed together. A small feature dimension can already contain nearly all posterior-relevant information; increasing polynomial degree can improve localization at fixed feature resolution; and accurately estimating the resulting moment geometry requires a pilot size that grows with basis complexity.  These are distinct terms in \eqref{eq:nevai-three-scale}; the numerical comparisons do not isolate their individual sizes.

Figure~\ref{fig:function-calibration} compares score values with the exact feature likelihood on 7000 randomly selected candidates from the same evaluation cloud. The score has a strongly compressed range and is far from the diagonal, despite useful ranking performance. Thus this experiment supports candidate screening at the stated resolution, rather than accurate pointwise likelihood reconstruction.
\begin{figure}[tbp]
\centering
\includegraphics[width=0.4\textwidth]{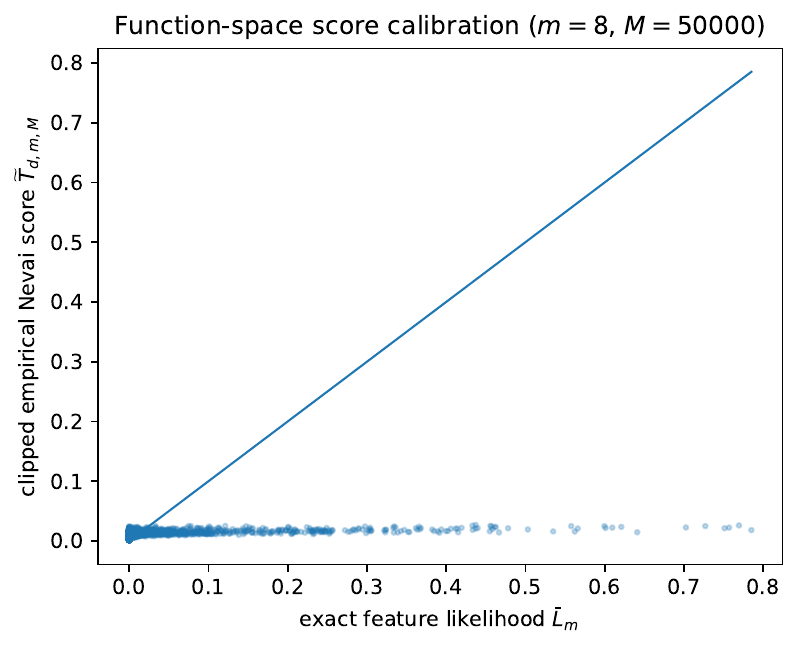}
\caption{Nevai score versus exact feature likelihood ($m=8$, $d=2$, $M=50000$). The line is the identity; score values are not rescaled.}
\label{fig:function-calibration}
\end{figure}
\FloatBarrier

\subsection{A diffusion-coefficient inverse problem}
\label{sec:num-diffusion}

We conclude with a two-parameter diffusion inverse problem using the same screening construction. Here $q=\theta\in\mathbb R^2$.  On \((0,1)\), consider
\[
-\frac{\mathrm d}{\mathrm d x}\!\left(a_\theta(x)\frac{\mathrm d u}{\mathrm d x}\right)=1,\qquad u(0)=u(1)=0.
\]
with
$a_\theta(x)=\exp\!\left\{0.9\theta_1\sin(\pi x)+0.7\theta_2\cos(2\pi x)\right\}$,   $\theta\in[-1,1]^2$.
The unknown enters the leading diffusion coefficient.  The deterministic forward map is evaluated by high-order Gauss--Legendre quadrature using the one-dimensional flux identity
$a_\theta(x)u_\theta'(x)=C_\theta-x$,  $C_\theta= \frac{\int_0^1 x/a_\theta(x)\,\dd x} {\int_0^1 1/a_\theta(x)\,\dd x}$.

We observe \(u_\theta\) at \(x=0.2,0.4,0.6,0.8\), use the uniform prior on \([-1,1]^2\), and take $\theta^\dagger=(0.55,-0.40)$, $\sigma=0.007$. The noise realization is fixed across all computations.  A weighted pilot of \(M=20000\) prior candidates uses the ordinary reduced Gaussian likelihood, and the score is formed from a degree-five total Legendre basis of size 21.  On an independent evaluation cloud of $N=60000$ prior draws, the Nevai score has Spearman rank correlation \(0.772\) with the exact reduced likelihood.  More importantly for screening, its top \(10\%\) and \(20\%\) prior fractions contain respectively \(88.9\%\) and \(98.3\%\) of the reference posterior mass.

\begin{figure}[t]
\centering
\includegraphics[width=0.98\textwidth]{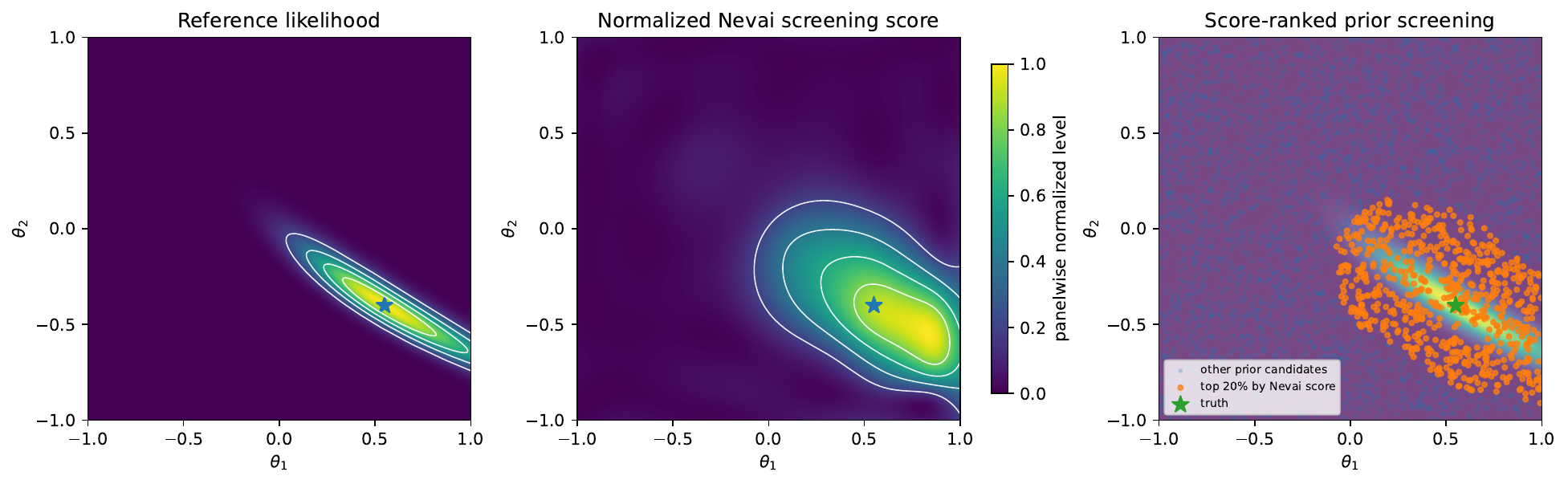}
\caption{Diffusion inversion: reference likelihood (left), Nevai score (center), and top-$20\%$ candidates (right). Field panels are independently rescaled to $[0,1]$.}
\label{fig:diffusion-screening}
\end{figure}

Figure~\ref{fig:diffusion-screening} uses separate min--max scaling for the likelihood and Nevai fields and common contour levels; candidate ranking uses the unrescaled score. It shows that the score is smoother than the likelihood. The reported capture values indicate that this smoothing still permits useful ranking in this example; they do not establish accurate likelihood reconstruction.  The example also illustrates why the Nevai score is attractive in concentrated posteriors: it uses the posterior-weighted moment matrix linearly and does not require inversion of an increasingly ill-conditioned matrix.

\section{Conclusion}
\label{sec:conclusion}

We have developed a posterior-moment view of Bayesian inverse problems in which a likelihood-weighted Hankel matrix records how the data reshape prior geometry, using ordinary deterministic likelihood evaluations at prior pilot draws.  Christoffel and Nevai constructions turn this matrix into complementary relevance scores: the former emphasizes inverse-moment contrast, while the latter is a bounded population average of the likelihood with direct matrix-perturbation control.  The construction uses no derivatives of the forward map.

For function-valued unknowns, nested feature maps link the finite-resolution construction to an exact feature posterior.  The conditional likelihood $L_m$ preserves the posterior information visible through the selected features, and the associated posterior converges to the full posterior as the filtration becomes complete.  The Hankel matrix is a finite-dimensional compression of the likelihood multiplication operator, while the empirical-Nevai estimate separates naturally into feature, polynomial-localization, and pilot errors.  The nonlinear function-coefficient PDE and the Gaussian inverse-source experiment illustrate the resulting tradeoff: a small feature set can already capture most of the posterior-relevant variation, but accurate screening still requires enough polynomial and pilot resolution.

The expanded nonlinear PDE comparison separates scalar-likelihood screening from forward-response approximation. Christoffel screening performs better than the tested direct likelihood regressions at small retention budgets, but the polynomial and kernel forward surrogates perform better still. The scores therefore provide an alternative when scalar likelihood information is the available interface; the experiments do not establish a general advantage over surrogate modelling. Pilot variability and conditional reference-bootstrap uncertainty are reported separately. Measured costs include fitting, feature construction, scoring, sorting, and the retained forward solves, and demonstrate savings only for the specified screening task.

The filtered Gaussian diagnostics clarify the role of smoothing. Increasing polynomial degree at a fixed filter approaches a smoothed likelihood, while feature selection can reduce error more than filtering at a fixed small degree. The independent-pilot construction is implemented on a controlled sparse-coordinate example. This supports the interpretation of the error terms, without claiming a Gaussian Christoffel convergence rate or broad adaptive performance on nonlinear PDEs.

Adequate prior coverage of posterior-relevant regions remains necessary in both the pilot and candidate pool. Moment construction cannot recover regions absent from the available information, and small evidence can make an absolute-error regret bound uninformative. The numerical PDE has a 24-mode coefficient representation; the experiments do not establish continuum accuracy, mesh-independent performance, or an end-to-end exact posterior sampling advantage. Extensions to adaptive pilots and corrected inference should be assessed separately from the screening guarantees proved here.

\appendix
\section{Proofs of selected results}
\label{app:proofs}
The proofs below complete the results stated in the main text.

\subsection{Matrix concentration}\label{app:matrix-concentration}
\begin{proof}[Proof of Theorem~\ref{thm:matrix-concentration}]
Set $\boldsymbol{\Phi}_i=\boldsymbol{\phi}_d(Q_i)$,
$X_i=L(Q_i; y)\boldsymbol{\Phi}_i\boldsymbol{\Phi}_i^\top$, and
$Y_i=X_i-H_d$. Then $\E Y_i=0$ and
$\widehat H_d^{(M)}-H_d=M^{-1}\sum_iY_i$.
Since $\|X_i\|_{\op}\le\kappa_d$ and $0\preceq H_d\preceq I$,
we have $\|Y_i\|_{\op}\le\kappa_d+1\le2\kappa_d$.
Moreover,
\[
 X_i^2=L(Q_i;y)^2\|\boldsymbol{\Phi}_i\|_2^2
 \boldsymbol{\Phi}_i\boldsymbol{\Phi}_i^\top\preceq\kappa_dX_i.
\]
Consequently,
\begin{align*}
\begin{aligned}
 &\E Y_i^2=\E X_i^2-H_d^2\preceq\kappa_dH_d\preceq\kappa_d I,\\
 &\left\|\sum_{i=1}^M\E Y_i^2\right\|_{\op}\le M\kappa_d.
\end{aligned}
\end{align*}
  The self-adjoint matrix Bernstein inequality \cite{Tropp2012},
applied to $\sum_iY_i$ at threshold $Mt$, proves
\eqref{eq:matrix-concentration}.
\end{proof}

\subsection{Legendre localization}\label{app:legendre-localization}
\begin{lemma}
\label{lem:legendre-diagonal}
For the uniform probability measure $\dd\nu_{\rm U}(t)=\tfrac12\,\dd t$
on $[-1,1]$ and its degree-$k$ orthonormal Legendre kernel,
\begin{equation}
 K_k^{\rm L}(t,t)\ge \frac{k+1}{2\pi},
 \qquad t\in[-1,1],\quad k\ge0.
 \label{eq:legendre-diagonal-lemma}
\end{equation}
\end{lemma}
\begin{proof}
Compare $\nu_{\rm U}$ with the Chebyshev probability measure
$\dd\nu_{\rm C}(t)=[\pi\sqrt{1-t^2}]^{-1}\dd t$.
Since $\dd\nu_{\rm U}/\dd\nu_{\rm C}\le\pi/2$, the variational
characterization of the Christoffel function implies
$K_k^{\rm L}(t, t)\ge(2/\pi)K_k^{\rm C}(t, t)$.
For $t=\cos\theta$ with $0<\theta<\pi$, the orthonormal Chebyshev basis satisfies
\[
\begin{aligned}
 K_k^{\rm C}(t,t)
 &=1+2\sum_{j=1}^k\cos^2(j\theta)\\
 &=k+\frac12+\frac{\sin((2k+1)\theta)}{2\sin\theta}.
\end{aligned}
\]
The endpoint values are understood by continuity.
By symmetry we may assume $0<\theta\le\pi/2$.
If the last term is negative, then $\theta\ge\pi/(2k+1)$.
Since $\sin\theta\ge2\theta/\pi$ on this interval, 
$$\frac{\sin((2k+1)\theta)}{2\sin\theta} \ge-\frac{2k+1}{4},$$
 whereas if it is nonnegative the same lower bound is immediate.
Consequently $$K_k^{\rm C}(t, t)\ge(2k+1)/4\ge(k+1)/4,$$
and the comparison proves \eqref{eq:legendre-diagonal-lemma}.
\end{proof}

\begin{proof}[Proof of Theorem~\ref{thm:legendre-localization}]
Using $\mathsf P_k(t)=\sqrt{2k+1}P_k(t)$, where $P_k(1)=1$ is the classical Legendre normalization, the three-term recurrence has off-diagonal coefficient $a_{k+1}=(k+1)/\sqrt{(2k+1)(2k+3)}$, with $a_{k+1}^2\le1/3$. For the total-degree kernel, telescoping the one-dimensional recurrence in coordinate $j$ leads to the following boundary identity, where $\mathbf{e}_j$ is the $j$-th coordinate unit vector:
\begin{equation*}
(\zeta_j-\xi_j)K_{d,m}(\bxi,\bzeta)
 =\sum_{|\alpha|=d}a_{\alpha_j+1}
 \left[p_\alpha(\bxi)p_{\alpha+\mathbf{e}_j}(\bzeta)
       -p_{\alpha+\mathbf{e}_j}(\bxi)p_\alpha(\bzeta)\right].
\end{equation*}
Orthogonality therefore yields
\begin{equation}
 \int (\zeta_j-\xi_j)^2K_{d,m}(\bxi,\bzeta)^2\dd\rho_m^0(\bzeta)
 =\sum_{|\alpha|=d}a_{\alpha_j+1}^2
 \left[p_\alpha(\bxi)^2+p_{\alpha+\mathbf{e}_j}(\bxi)^2\right].
 \label{eq:boundary-square}
\end{equation}

It remains to control the diagonal denominator.  If $k=\lfloor d/m\rfloor$, then the tensor index set $\{0, \ldots, k\}^m$ is contained in the total-degree index set.  Hence $\mathcal K_{d, m}(\bxi) \ge\prod_{j=1}^m K_k^{\rm L}(\xi_j,\xi_j)$. By Lemma~\ref{lem:legendre-diagonal}, the one-dimensional Legendre kernel satisfies
$K_k^{\rm L}(t, t)\ge\frac{k+1}{2\pi}$,  $-1\le t\le1$,
so
\begin{equation}
  \mathcal K_{d,m}(\bxi)
  \ge
  \left(\frac{\lfloor d/m\rfloor+1}{2\pi}\right)^m.
  \label{eq:total-degree-diagonal-lower}
\end{equation}
Summing \eqref{eq:boundary-square} over $j$, dividing by \eqref{eq:total-degree-diagonal-lower}, integrating in $\bxi$, and using
$\int p_\alpha^2\dd\rho_m^0=1$ and
$\#\{\alpha:|\alpha|=d\}=\binom{d+m-1}{m-1}$ establishes
\eqref{eq:Vdm-exact-bound}.

For $d\ge m$, 
\begin{align*}
\binom{d+m-1}{m-1} \le\frac{(2d)^{m-1}}{(m-1)!},\quad \lfloor d/m\rfloor+1\ge d/m, 
\end{align*}
which proves \eqref{eq:Vdm-Cm} after using $d^{-1}\le2(d+1)^{-1}$.
Finally, the Nevai representation implies $|T_{d, m}(\bxi)-\bar L_m(\bxi)|^2 \le \ell_m^2\mathfrak m_{2, d, m}(\bxi)$, and integration proves \eqref{eq:nevai-legendre-rate}.
\end{proof}

\subsection{Comparison of the scores}\label{app:score-discrepancy}
Let $\mathbf{u}_{\bxi}:=\frac{\boldsymbol p_{d,m}(\bxi)}{\sqrt{\mathcal K_{d,m}(\bxi)}}$ and complete $\mathbf{u}_{\bxi}$ to an orthogonal matrix $U_{\bxi}=[\mathbf{u}_{\bxi},U_{\bxi}^\perp]$.  Write
\begin{equation}
 U_{\bxi}^\top H_{d,m}U_{\bxi}
 =\begin{pmatrix}
 T_{d,m}(\bxi)&\mathbf{b}_{\bxi}^\top\\
 \mathbf{b}_{\bxi}&C_{\bxi}
 \end{pmatrix}.
 \label{eq:score-block-matrix}
\end{equation}
The Schur complement identity reads
\begin{equation}
  T_{d,m}(\bxi)-R_{d,m}(\bxi)
  =\mathbf{b}_{\bxi}^\top C_{\bxi}^{-1}\mathbf{b}_{\bxi}\ge0,
  \label{eq:schur-score-gap}
\end{equation}

\begin{proof}[Proof of Proposition~\ref{thm:score-discrepancy}]
Condition \eqref{eq:positive-feature-likelihood} implies
$a_mI\preceq H_{d,m}\preceq I$, hence $C_{\bxi}\succeq a_mI$.  The top-left block of the inverse of \eqref{eq:score-block-matrix} is the reciprocal of the Schur complement, so $R_{d,m}(\bxi) =T_{d,m}(\bxi)-\mathbf{b}_{\bxi}^\top C_{\bxi}^{-1}\mathbf{b}_{\bxi}$, which proves \eqref{eq:schur-score-gap}.

Let $k_{\bxi}(\bzeta):=\frac{K_{d, m}(\bxi,\bzeta)}{\sqrt{\mathcal K_{d, m}(\bxi)}}$. Then $\|k_{\bxi}\|_{L^2(\rho_m^0)}=1$ and
$T_{d,m}(\bxi)=\int\bar L_m(\bzeta)k_{\bxi}(\bzeta)^2\dd\rho_m^0(\bzeta)$.  The vector $\mathbf{b}_{\bxi}$ is the coefficient vector of the component orthogonal to $k_{\bxi}$ after projecting $\bar L_m k_{\bxi}$ onto the degree-$d$ polynomial space in $L^2(\rho_m^0)$. Therefore $\|\mathbf{b}_{\bxi}\|_2^2 \le \|(\bar L_m-T_{d,m}(\bxi))k_{\bxi}\|_{L^2(\rho_m^0)}^2 =\operatorname{Var}_{\pi_{d,m,\bxi}}(\bar L_m)$. Since $C_{\bxi}^{-1}\preceq a_m^{-1}I$, this proves
\eqref{eq:variance-score-gap}.  The Lipschitz estimate follows by bounding the variance by
$\int(\bar L_m(\bzeta)-\bar L_m(\bxi))^2\dd\pi_{d,m,\bxi}(\bzeta)\le\ell_m^2\mathfrak m_{2,d,m}(\bxi)$.
\end{proof}
\subsection{Filtered Gaussian localization}\label{app:gaussian-localized}
\begin{proof}[Proof of Theorem~\ref{thm:gaussian-localized}]
Let $K_{\boldsymbol{\rho}}$ denote the infinite Mehler kernel, and let $T_{\boldsymbol{\rho},S}$ be its normalized squared-kernel average of $f_S$. Let $\eta_{\bxi}$ be Gaussian with mean $(b_j\xi_j)_{j\in S}$ and covariance $\operatorname{diag}(v_j)$.
By orthogonality and the Mehler formula,
\[
 \frac{K_{\boldsymbol{\rho}}(\bxi,\bzeta)^2\,\dd\gamma_{\mathrm G,S}(\bzeta)}
      {\int K_{\boldsymbol{\rho}}(\bxi,\bzeta)^2\,\dd\gamma_{\mathrm G,S}(\bzeta)}
 =\dd\eta_{\bxi}(\bzeta).
\]
Let $\mathbf G_1,\mathbf G_2$ be independent with law $\gamma_{\mathrm G,S}$, and write $B_0=\operatorname{diag}(b_j)$ and $V=\operatorname{diag}(v_j)$. By \eqref{eq:weighted-feature-lipschitz} and Jensen's inequality,
\[
\begin{aligned}
 \|T_{\boldsymbol\rho,S}-f_S\|_{L^2(\gamma_{\mathrm G,S})}
 &\le\bigl\{\E\|B_S[(B_0-I)\mathbf G_1+V^{1/2}\mathbf G_2]\|_{\mathcal Y}^2\bigr\}^{1/2}\\
 &=\left[\sum_{j\in S}\|B_S\mathbf e_j\|_{\mathcal Y}^2((1-b_j)^2+v_j)\right]^{1/2}.
\end{aligned}
\]

For completeness, fix $\bxi$, write $a=K_{\boldsymbol{\rho},\Lambda}(\bxi,\cdot)$,
$b=K_{\boldsymbol{\rho}}(\bxi,\cdot)$ in $L^2(\gamma_{\mathrm G,S})$, and put
$F=\|a\|_2^2$, $G=\|b\|_2^2$.  Orthogonal projection onto the selected
Hermite indices implies $\langle a, b\rangle=F$ and
$G-F=\sum_{\alpha\notin\Lambda}\boldsymbol{\rho}^{2\alpha}h_\alpha(\bxi)^2$.
The Hellinger affinity of the probability densities $a^2/F$ and $b^2/G$
is at least $|\langle a, b\rangle|/\sqrt{FG}=\sqrt{F/G}$.  For probability
densities $p, q$, Cauchy--Schwarz applied to
$|p-q|=|\sqrt p-\sqrt q|(\sqrt p+\sqrt q)$ implies
$\|P-Q\|_{\rm TV}\le\sqrt{1-(\int\sqrt{p q})^2}$.  Hence the two
kernel probabilities have total-variation distance at most
$\sqrt{1-F/G}\le\sqrt{G-F}$, since $G\ge F\ge1$.
The averages of any $[0,1]$-valued function differ by at most that distance.
After squaring and integrating over $\bxi\sim\gamma_{\mathrm G,S}$, we obtain
\[
 \|T_{\boldsymbol{\rho},\Lambda,S}-T_{\boldsymbol{\rho},S}\|_2^2
 \le\sum_{\alpha\notin\Lambda}\boldsymbol{\rho}^{2\alpha}
 =\mathcal A_\Lambda(\boldsymbol{\rho}).
\]
The triangle inequality completes the proof.
\end{proof}

\subsection{Independent-pilot feature selection}\label{app:adaptive-hankel-nevai}
\begin{proof}[Proof of Theorem~\ref{thm:adaptive-hankel-nevai}]
Let $\mathbf{c}_\Gamma=(c_\alpha)_{\alpha\in\Gamma}$, $\widehat{\mathbf{c}}_\Gamma=(\widehat c_\alpha)_{\alpha\in\Gamma}$, and
$\mathbf{e}_\Gamma=\widehat{\mathbf{c}}_\Gamma-\mathbf{c}_\Gamma$.
Since $0\le L\le1$ and the Hermite basis is orthonormal,
$\E\|\mathbf{e}_\Gamma\|_2^2\le |\Gamma|/M_{\rm sel}$.
Let $\mathsf Q_S$ be the coordinate projection of candidate coefficients onto
indices with $\operatorname{supp}\alpha\subseteq S$.
The minimizing property of $\widehat S$ and the triangle inequality yield,
for any admissible $S$, 
\[
\|(I-\mathsf Q_{\widehat S})\mathbf{c}_\Gamma\|_2 \le\|(I-\mathsf Q_S)\mathbf{c}_\Gamma\|_2+2\|\mathbf{e}_\Gamma\|_2.
\]
 Since $\eta_S^2=\|(I-\mathsf Q_S)\mathbf{c}_\Gamma\|_2^2+
 \sum_{\alpha\notin\Gamma,\operatorname{supp}\alpha\not\subseteq S}c_\alpha^2$,
we have
$\eta_{\widehat S}\le\|(I-\mathsf Q_{\widehat S})\mathbf{c}_\Gamma\|_2
 +\varepsilon_\Gamma$ and
$\|(I-\mathsf Q_S)\mathbf{c}_\Gamma\|_2\le\eta_S$.
Take infima and expectations to obtain \eqref{eq:feature-selection-bound}.

For the estimation step, condition on the first pilot.  For each feature
value $\bxi$, the normalized filtered kernel is a polynomial $p_{\bxi}$ of total
degree at most $D$ with $\|p_{\bxi}\|_{L^2(\gamma_{\mathrm G,\widehat S})}=1$.
Gaussian hypercontractivity \cite{Janson1997} yields
$\|p_{\bxi}\|_4\le3^{D/2}$.
The variance of its independently sampled weighted-square mean is thus
bounded by $9^D/M_{\rm est}$, uniformly in $\bxi$.
Integrating over $\bxi$ and applying Jensen's inequality, we obtain 
\[
\E_{\rm est}\|\widehat T_{\boldsymbol{\rho},\Lambda,\widehat S} -T_{\boldsymbol{\rho},\Lambda,\widehat S}\|_2 \le 3^D/\sqrt{M_{\rm est}}.
\]
 The feature, localization, and truncation terms are controlled by
\eqref{eq:feature-selection-energy} and
Theorem~\ref{thm:gaussian-localized}.  Apply the triangle inequality,
then average over the selection pilot.
\end{proof}

\begingroup\small\sloppy
\bibliographystyle{plain}
\bibliography{references}
\endgroup
\end{document}